\documentclass[%
 aip,
 jmp,%
 amsmath,amssymb,
preprint,
]{revtex4-1}
\usepackage{graphicx}
\usepackage{dcolumn}
\usepackage{bm}
\usepackage{amsmath}
\usepackage{amsthm}
\usepackage{amsfonts}       
\usepackage{amssymb}
\usepackage{color}

\newtheorem{theorem}{Theorem}
\newtheorem{lemma}{Lemma}
\newtheorem{coro}{Corollary}
\newtheorem{defi}{Definition}
\newtheorem{prop}{Proposition}

\theoremstyle{definition}
\newtheorem{remark}[defi]{Remark}

\begin{document}


\title[Unitary Dilations of Discrete-Time Quantum-Dynamical Semigroups]{Unitary Dilations of Discrete-Time Quantum-Dynamical Semigroups}

\author{Frederik \surname{vom Ende}}\thanks{(corresponding author)}\email[E-mail: ]{frederik.vom-ende@tum.de}
\affiliation{ 
Technische Universit{\"a}t M{\"u}nchen, 
	Dept. Chem., Lichtenbergstra{\ss}e 4, 85747 Garching, Germany}
	\affiliation{   Munich Centre for Quantum Science and Technology (MCQST),  Schellingstraße 4, 80799~M{\"u}nchen, Germany
}

\author{Gunther \surname{Dirr}}\email[E-mail: ]{dirr@mathematik.uni-wuerzburg.de}
\affiliation{%
Universit{\"a}t W{\"u}rzburg, 
        Institut f{\"u}r Mathematik, Emil-Fischer-Stra{\ss}e 40, 97074 W{\"u}rzburg, Germany
}%

\date{\today}

\begin{abstract}
We show that the discrete-time evolution of an open quantum system generated by a single quantum channel $T$ 
can be embedded in the discrete-time evolution of an enlarged closed quantum system, i.e.~we construct
a unitary dilation of the discrete-time quantum-dynamical semigroup $(T^n)_{n \in \mathbb N_0}$. In the case of a 
cyclic channel $T$, the auxiliary space may be chosen (partially) finite-dimensional. We further investigate 
discrete-time quantum control systems generated by finitely many commuting quantum channels and prove
a similar unitary dilation result as in the case of a single channel.
\end{abstract}



\maketitle
\section{Introduction}

%

Stimulated by the seminal work of Arveson\cite{arveson1972}, Lindblad\cite{Lindblad76}, Gorini, Kossakowski
and Sudarshan\cite{GKS76} in the mid 1970s many efforts have been made to obtain dilation results of various
degrees of generality for semigroups of completely positive operators.

For instance, Davies\cite[Ch.~9, Thm.~4.3]{Davies} proved that for any continuous semigroup $(T_t)_{t\in\mathbb R^+_0}$
of completely positive, unital operators acting on a finite-dimensional Hilbert spaces $\mathcal H$
(or, more precisely, on the corresponding $W^*$-algebra $\mathcal B(\mathcal H)$
of all bounded linear operators) there exists a Hilbert space $\mathcal K$, a pure state
$\omega$ in $\mathcal B(\mathcal K)$, and a strongly continuous one-parameter group $(U_t)_{t\in\mathbb R}$
of unitaries on $\mathcal H\otimes\mathcal K$ such that
\begin{align}\label{eq:0}
T_t(A)=\operatorname{tr}_{\omega}(U_t^\dagger(A\otimes\operatorname{id}_{\mathcal K})U_t)
\end{align}
holds for all $A \in \mathcal B(\mathcal H)$ and $t\in\mathbb R^+_0$. For infinite-dimensional
$\mathcal H$, there is a whole zoo of similar results. While Davies \cite{davies78,Davies79}, 
Evans\cite{Evans1976groups}, and Evans \& Lewis\cite{EvansLewis1976,EvansLewisOverview77} focused
primarily on one-parameter semigroups $(T_t)_{t \in\mathbb R^+_0}$ of different continuity type, 
K\"ummerer\cite{Kuemmerer83} discussed at great length the discrete-time case $(T^n)_{n \in \mathbb N}$. 
However, to the best of our knowledge, for arbitrary Hilbert spaces a dilation result 
of the above form \eqref{eq:0} is still not available. 

In the following, we give a short chronological overview on those contributions which are relevant
and closely related to our work. Further results and a brief survey over the latest developments  
can be found in \cite{Bhat1996, bhat_skeide_2000,muhly_solel_2002, Gaebler2015} and \cite{sawada2018}.
For the readers' convenience we collected some standard terminology and basic results on 
dilations and (completely) positive maps, which are well known to experts in this area, in the glossary 
of Appendix \ref{sec:ditalions}.\medskip

In \cite[Thm.1 and Thm.2]{Evans1976groups}, Evans shows that every family $(T_g)_{g \in G}$ of completely 
positive, unital operators acting on a unital $C^*$-algebra $\mathcal A$ and indexed by an arbitrary
group $G$ admits a unitary dilation, i.e.
\begin{align*}
T_g(A) = E(U_g^\dagger J(A)U_g)
\end{align*}
for all $A \in \mathcal A$ and $g \in G$, where $(U_g)_{g \in G}$ is a unitary representation of $G$
on some Hilbert space $\mathcal K$ and $E$ a conditional expectation with corresponding injection $J$ into $\mathcal B(\mathcal K)$. Remarkably, he need not assume that $g \mapsto T_g$ 
is a group homomorphism. 
His result can be regarded as $C^*$-counterpart to Sz.-Nagy's \cite{nagy1953,nagy1970} and Stroescu's 
\cite{stroescu1973} work on isometric dilations on Hilbert and Banach spaces, respectively. While 
possible generalizations to $W^*$-algebras are addressed by Evans, continuity issues of the map 
$g \mapsto U_g$ are disregarded completely. His proof is based on Stinespring's representation 
$T_g(x) = V^*_g \pi_g(x) V_g$ which of course exists for all $g \in G$. However, he did not exploit
the fact that one can choose a common Hilbert space for all $\pi_g$ which leads to a substantial
simplification in our approach.
In \cite[Thm.2]{EvansLewis1976} Evans \& Lewis focus on norm-continuous semigroups $(T_t)_{t\in\mathbb R^+_0}$
of ultraweakly continuous, completely positive and unital operators acting on a separable
Hilbert space $\mathcal H$. They obtain a unitary dilation 
\begin{align*}
T_t(A) = E(U_t^\dagger J(A)U_t)\,,
\end{align*}
for all $A \in \mathcal B(\mathcal H)$ and $t \in \mathbb R^+_0$, where $(U_t)_{t\in\mathbb R}$ is a
strongly continuous group of unitary operators acting on some extended Hilbert space $\mathcal K$
(and $E, J$ as above). Their proof exploits the fact that the explicit form of the
infinitesimal generator of $(T_t)_{t\in\mathbb R^+_0}$ is well-known due to the work of
Lindblad\cite{Lindblad76}.
In \cite{EvansLewisOverview77} Evans \& Lewis provide an overview on dilation results known at that time
including some minor generalisations of their previous work \cite{EvansLewis1976}. 

\medskip
For locally compact groups $G$, Davies \cite[Thm.2.1 and Thm.3.1]{davies78} obtains the following 
rather general result: 
Let $(T_g)_{g\in G}$ be a strongly continuous family of ultraweakly continuous, completely positive
and unital 
operators on $\mathcal B(\mathcal H)$. Then there exists
a Hilbert space $\mathcal K$, a strongly continuous unitary representation $U$ of $G$ on 
$\mathcal H \otimes \mathcal K$ and conditional expectations 
$E_n:\mathcal B(\mathcal H\otimes\mathcal K)\to\mathcal B(\mathcal H)$ (for all $n \in \mathbb N$)
such that
\begin{align}\label{eq:0_1}
T_g(A)=\lim_{n\to\infty} E_n(U_g^\dagger(A\otimes\operatorname{id}_{\mathcal K})U_g)
\end{align}
holds for all $A \in \mathcal B(\mathcal H)$ and all $g\in G$ in the weak operator topology. Here, 
$E_n$ is of the form $E_n(A):=V_n^\dagger AV_n$ where $V_n:\mathcal H\to\mathcal H\otimes\mathcal K$
are isometric embeddings. This seems to be the result which is closest to (\ref{eq:0}) in infinite 
dimensions, but it is not known whether the limit in (\ref{eq:0_1}) is
necessary or not\cite[cf.~p.~335]{davies78}. 
For discrete-time systems $T_n := T^n$, $n \in \mathbb N$ or, more accurately, for an appropriate
extension to $G = \mathbb Z$ Davies' approach and ours are quite similar---in particular due to the 
fact that in this case, the limit in \eqref{eq:0_1} can be avoided as $G$ is discrete. More precisely, 
Davies first extends the state space from 
$\mathcal H$ to $L^2(\mathbb Z, \mathcal H) \cong \ell_2(\mathbb Z) \otimes \mathcal H$ 
such that $(T_n)_{n \in \mathbb Z}$ can be regarded as one completely positive, unital operator from 
$\mathcal B(\mathcal H)$ to $\mathcal B(L^2(\mathbb Z, \mathcal H))$. He then applies Stinespring's 
representation theorem to obtain an dilation of $(T_n)_{n \in \mathbb Z}$ on a larger state space
$L^2(\mathbb Z, \mathcal H) \otimes \mathcal{K}$. We, however, exploit Stinespring's
result first to guarantee for all $n \in \mathbb N$ a dilation of the form 
\begin{align*}
T^n(A)=\operatorname{tr}_{\omega}((U^\dagger)^n(A\otimes\operatorname{id}_{\mathcal K})U^n)\,,
\end{align*}
where $\omega$ and $\mathcal K$ are independent of $n \in \mathbb N$, to then enlarge the state space to 
$\ell_2(\mathbb Z) \otimes \mathcal H \otimes \mathcal K \cong L^2(\mathbb Z, \mathcal H \otimes \mathcal K)$. 
Although both approaches differ only in the order of the construction steps the resulting dilations
behave quite differently: While Davies' construction is more ``flexible'' as one can see, e.g., in Section
\ref{section_control}, Remark \ref{ch_4_bem_7}.2, ours yields the desired partial trace
structure of \eqref{eq:0} which is
in general not satisfied for \eqref{eq:0_1} even if the limit can be avoided\cite{Note1}
.\medskip

K\"ummerer\cite{Kuemmerer83} discussed the discrete-time case $(T^n)_{n \in \mathbb N}$ in detail.
However, his setting significantly differs from ours. In his sense, a discrete-time quantum dynamical
system consists of a triple $(\mathcal A,\varphi,T)$, where $\mathcal A \subset \mathcal B(\mathcal H)$ 
is a $W^*$-algebra, $T$ an ultraweakly continuous, 
completely positive and unital operator which acts on $\mathcal A$ and leaves a faithful normal
state $\varphi \in \mathcal A^*$ invariant, i.e.~$\varphi\circ T = \varphi$. The latter condition can 
be thermodynamically motivated as $\varphi$ can be interpreted as an equilibrium state which is preserved 
under composition with $T$ and every power of it. This constraint on the quantum channel $T$ obviously 
narrows down the possible choices of $T$. Even more restrictive is K\"ummerer's definition of a first
order dilation of $(\mathcal A,\varphi,T)$. Here, he requires the existence of a reversible quantum 
dynamical system $(\mathfrak A,\hat{\varphi},\hat{T})$, i.e.~$\hat{T}$ is a $*$-automorphism on $\mathfrak A$ and $E$ is a conditional expectation with corresponding injection $J$ such that 
\begin{align*}
T(A) = E\big(\hat{T}(J(A))\big)\qquad \text{and} \qquad \varphi\circ E = \hat{\varphi}\,.
\end{align*}
for all $A \in \mathcal A$. In doing so, the condition $\varphi\circ E = \hat{\varphi}$ is the delicate
part. For instance, the standard Kraus/Stinspring representation which constitutes a (first order) dilation
does in general not satisfy this condition---note that, by definition, $\hat{\varphi}$ has to be a faithful 
normal state---and therefore even first order dilations in K\"ummerer's sense need not exist as the existence
of a $\varphi$-adjoint is not guaranteed, cf.~\cite[Prop.~2.1.8 ff.]{Kuemmerer83}. Within his 
setting, K\"ummerer proved (cf.~Thm 4.2.1, Cor 4.2.3) that a quantum dynamical system
$(\mathcal B(\mathcal H),\varphi,T)$ has a dilation of first order if and only if it admits a Markovian one 
of first order which in turn implies that $(\mathcal B(\mathcal H),\varphi,T)$ also allows a Markovian dilation 
of arbitrary order. His definition of Markovianity can be regarded as a $W^*$-algebra counterpart of a 
well-known subspace condition\cite{Note2} 
which guarantees for contractions on Hilbert spaces that a first order unitary dilation 
$T = P_{\mathcal H}U|_{\mathcal H}$ is already a dilation (of arbitrary order), 
i.e.~$T^n = P_{\mathcal H}U^n|_{\mathcal H}$ holds for all $n \in \mathbb N$. To achieve a Markovian dilation
he imbedded the given $W^*$-algebra $\mathcal A = \mathcal B(\mathcal H)$ in an infinite product/sum of 
$W^*$-algebras. Our approach considerably deviates from his construction since we use first order 
Stinespring/Kraus dilations for $T^n$ which of course exist for all $n \in \mathbb N$ but in general 
do not satisfy K\"ummerer's faithful state condition.\medskip

Probably one of the strongest semigroup dilation results so far was presented
by Gaebler\cite[Thm.~5.10]{Gaebler2015}. Using Sauvageot's theory he showed that for a norm-continuous
semigroup $( T_t)_{t\in\mathbb R^+}$ of ultraweakly continuous, completely positive and unital operators
acting on a $W^*$-algebra $\mathcal A$ with separable pre-dual, there exists a unital dilation 
$(\mathfrak A,(\sigma_t)_{t\in\mathbb R^+_0},J,E)$ of $( T_t)_{t\in\mathbb R^+_0}$ 
(cf.~Def.~\ref{def:dilation-von-Neumann-algebras}) where $\mathfrak A$ has separable pre-dual and 
$((\sigma_t)_{t\in\mathbb R^+_0},J,E)$ satisfies the strong dilation property, i.e.~$T_t\circ E=E\circ\sigma_t$
for all $t\in\mathbb R^+_0$. The strength of this result, however, comes at the cost of lacking any partial
trace structure of the form (\ref{eq:0}).\medskip

The paper is organized as follows: After some preliminaries on trace-class operators and quantum channels,
we present our main results in Section \ref{sec:main_results}: (i) For discrete-time quantum-dynamical semigroups
on separable Hilbert spaces, a unitary dilation of the form (\ref{eq:0}) is proved. (ii) If the semigroup
in question is generated by a cyclic quantum channel, then the auxiliary Hilbert space can be chosen partially
finite-dimensional. (iii) Finally, for discrete-time quantum control systems, the control of which can be 
switched between a finite number of commuting channels, a unitary dilation of the form (\ref{eq:0}) is derived.

\section{Preliminaries and Notation}\label{sec:perlim}

In this section, we fix our notation and recall some basic material on Schr\"odinger and
Heisenberg quantum channels. These results should be known to experts in this area.

Henceforth, let $\mathcal G, \mathcal H$ 
be infinite-dimensional separable complex Hilbert spaces and $\mathcal X, \mathcal Y$ 
real or complex Banach spaces. By convention, all scalar products on complex Hilbert spaces are assumed to be conjugate linear in the first argument and linear in the second. Moreover, let $\mathcal B(\mathcal G), \mathcal B(\mathcal H)$ 
denote the set of all bounded operators acting on $\mathcal G, \mathcal H$ 
and let $\mathcal B(\mathcal X), \mathcal B(\mathcal Y)$ 
be defined respectively.

Recall that an operator $A\in\mathcal B(\mathcal H)$ on a complex Hilbert space is said to be positive semi-definite, denoted by $A\geq 0$, iff $\langle x,Ax\rangle\geq 0$ for all $x\in\mathcal H$. Because we consider complex Hilbert spaces, $A\geq 0$ directly implies that $A$ is self-adjoint via the polarization identity, cf.~\cite[Prop.~2.4.6]{kadisonringrose1983}---else, self-adjointness would have to be required in the definition of $A\geq 0$.

\subsection{Quantum Channels}
\label{sec:qu-channels}

Let $\mathcal B^1(\mathcal H) \subset \mathcal B(\mathcal H)$ be the subset of all \emph{trace-class
operators}, i.e.~$\mathcal B^1(\mathcal H)$ is the largest subspace of $\mathcal B(\mathcal H)$
which allows to define the \emph{trace} of an operator $A$ via
\begin{align}\label{def_trace_eq}
\operatorname{tr}(A):=\sum\nolimits_{i\in I}\langle e_i,Ae_i\rangle
\end{align}
such that the right-hand side of \eqref{def_trace_eq} is finite and independent of the choice
of the orthonormal basis $(e_i)_{i \in I}$. More precisely, $\mathcal B^1(\mathcal H)$ can be defined
either as the set of all compact operators $A \in \mathcal B(\mathcal H)$ whose singular values
$\sigma_n(A)$ are summable, i.e.~
\begin{equation}\label{def1_trace_eq}
\nu_1(A):=\sum\nolimits_{n\in \mathbb{N}} \sigma_n(A) < \infty
\end{equation}
or, equivalently\cite[Thm.VI.21]{ReedSimon1}, as the set of all $A \in \mathcal B(\mathcal H)$ such that 
\begin{align}\label{def2_trace_eq}
\sum\nolimits_{i\in I}\langle e_i, \sqrt{A^\dagger A} e_i\rangle < \infty
\end{align}
is summable for some orthonormal basis $(e_i)_{i \in I}$ of $\mathcal H$.

Because of $\sqrt{A^\dagger A}\geq 0$, all summands in \eqref{def2_trace_eq} are non-negative and 
therefore\cite[Thm.VI.18]{ReedSimon1}, the value of the left-hand side of \eqref{def2_trace_eq} is
independent of the chosen orthonormal basis. Moreover, one has $\operatorname{tr}(\sqrt{A^\dagger A})=\nu_1(A)$
for all $A\in\mathcal B^1(\mathcal H)$ which readily implies $\operatorname{tr}(A)=\nu_1(A)$ if $A\geq 0$.
Finally, we note that for finite-dimensional Hilbert spaces the sets $\mathcal B(\mathcal H)$ and 
$\mathcal B^1(\mathcal H)$ coincide with the set of all linear operators acting on $\mathcal H$ and that 
for arbitrary Hilbert spaces, $\mathcal B^1(\mathcal H)$ constitutes a Banach space with respect to
the \emph{trace norm} $\nu_1$ given by \eqref{def1_trace_eq}. For more on these topics we refer
to \cite[Ch.VI.6]{ReedSimon1} and \cite[Ch.16]{MeiseVogt}.\medskip

An operator $\rho\in\mathcal B^1(\mathcal H)$ which is positive semi-definite and fulfills
$\operatorname{tr}(\rho)=1$ is called a \emph{state} and the set of all states is denoted by
\begin{align*}
\mathbb D(\mathcal H):=\lbrace \rho\in\mathcal B^1(\mathcal H)\,|\,\rho\text{ is state}\rbrace\,.
\end{align*}
A state $\rho$ is said to be \emph{pure} if it has rank one. Certainly, the corresponding definitions
apply to $\mathcal B^1(\mathcal G)$ and $\mathbb D(\mathcal G)$. After these preliminaries, we can 
introduce the key concepts.
\begin{defi}\label{def1}
\begin{itemize}
\item[(a)] A linear map $T:\mathcal B^1(\mathcal H)\to\mathcal B^1(\mathcal G)$ is said to be positive if $T(A)\geq 0$ for all positive semi-definite $A\in \mathcal B^1(\mathcal H)$.
\item[(b)] A linear map $T:\mathcal B^1(\mathcal H)\to\mathcal B^1(\mathcal G)$ is said to be completely positive
if for all $n\in\mathbb N$ the maps $T\otimes\operatorname{id}_n: \mathcal B^1(\mathcal H \otimes \mathbb{C}^n)\to\mathcal B^1(\mathcal G \otimes \mathbb{C}^n)$
are positive.
\item[(c)] A Schr\"odinger quantum channel is a linear, completely positive and trace-preserving map
$T:\mathcal B^1(\mathcal H)\to\mathcal B^1(\mathcal G)$. Furthermore, we define
\begin{align*}
Q_S (\mathcal H,\mathcal G):=
\lbrace T:\mathcal B^1(\mathcal H)\to\mathcal B^1(\mathcal G)\, |\, T\text{ is Schr\"odinger quantum channel}\,\rbrace
\end{align*}
and $Q_S (\mathcal H):= Q_S (\mathcal H,\mathcal H)$.
\end{itemize}
\end{defi}
\noindent Note that Definition \ref{def1} (a) and (b) also make sense for maps from 
$\mathcal B(\mathcal H) \to \mathcal B(\mathcal G)$ instead of $\mathcal B^1(\mathcal H)\to\mathcal B^1(\mathcal G)$.
\medskip

Clearly, every positive, trace-preserving map and thus every Schr\"odinger quantum channel maps states
to states. Further algebraic and topological properties of $Q_S (\mathcal H)$ which are crucial in the 
following are summarized in the following theorem, the proof of which can be found in Appendix \ref{app:monoid}. 

\begin{theorem}\label{thm_monoid}
The set $Q_S (\mathcal H)$ is a convex subsemigroup of $\mathcal B(\mathcal B^1(\mathcal H))$ with unity
element $\operatorname{id}_{\mathcal B^1(\mathcal H)}$. Moreover, $Q_S(\mathcal H)$ is closed in
$\mathcal B(\mathcal B^1(\mathcal H))$ with respect to the weak operator, strong operator and uniform
operator topology.
\end{theorem}

\begin{remark}\label{rem:boundedness}
\begin{enumerate}
\item
Here one should emphasize that it is not necessary to require boundedness of Schr\"odinger quantum channels. In fact, one can easily prove that any positive linear map is bounded automatically
\cite[Ch.~2, Lemma~2.1]{Davies}.
\item
In Proposition \ref{thm_q_norm_1} we will see that $Q_S (\mathcal H)$ is actually a closed
convex subset of the unit sphere of $\mathcal B(\mathcal B^1(\mathcal H))$.
Note that the existence of non-trivial convex subsets on the unit sphere of 
$\mathcal B(\mathcal B^1(\mathcal H))$ is a consequence of the non-strict convexity of the operator norm.
\end{enumerate}
\end{remark}

The following beautiful and well-known representation result for Schr\"odinger quantum channels which
can be traced back to Kraus is the starting of our work.

\begin{theorem}\label{thm1}
For every $T\in Q_S(\mathcal H)$ there exists a separable Hilbert space $\mathcal K$, a pure
$\omega\in\mathbb D(\mathcal K)$ and a unitary $U\in\mathcal B(\mathcal H\otimes\mathcal K)$ such that
\begin{align}\label{eq:T_stinespring}
T(A)=\operatorname{tr}_{\mathcal K}(U(A\otimes \omega)U^\dagger)
\end{align}
for all $A\in\mathcal B^1(\mathcal H)$.
\end{theorem}

\noindent
Here $\operatorname{tr}_{\mathcal K}:\mathcal B^1(\mathcal H\otimes\mathcal K)\to\mathcal B^1(\mathcal H)$
is the partial trace \textit{with respect to the Hilbert space} $\mathcal K$ which is defined via
\begin{align}\label{eq:partial_trace_1}
\operatorname{tr}(B\operatorname{tr}_{\mathcal K}(A))=\operatorname{tr}((B\otimes\operatorname{id}_{\mathcal K})A)
\end{align}
for all $B\in\mathcal B(\mathcal H)$ and all $A\in\mathcal B^1(\mathcal H\otimes\mathcal K)$.

For a complete proof of Theorem \ref{thm1}, see \cite[second part of Thm.2]{Kraus}. Here, we only
emphasize that the separable auxiliary space $\mathcal K$ can be chosen independently of $T$; for 
instance, $\mathcal K := \ell_2(\mathbb N)$ constitutes such a universal auxiliary space. Moreover,
once $\mathcal K$ is fixed, $\omega\in\mathbb D(\mathcal K)$ can be chosen as any orthogonal rank-$1$
projection. Thus $\omega$ is pure and independent of $T$, too.

\begin{coro}[General Stinespring Dilation]\label{coro_gen_stinespring}
For every $T\in Q_S(\mathcal H,\mathcal G)$ there exists a separable Hilbert space $\mathcal K$, pure $\omega_G\in\mathbb D(\mathcal G)$, $\omega_K\in\mathbb D(\mathcal K)$ and a unitary 
$U\in\mathcal B(\mathcal H\otimes\mathcal G\otimes\mathcal K)$ such that
\begin{align*}
T(A)=(\operatorname{tr}_{\mathcal H}\circ\operatorname{tr}_{\mathcal K})(U(A\otimes \omega_G\otimes\omega_K)U^\dagger)
\end{align*}
for all $A\in\mathcal B^1(\mathcal H)$.
\end{coro}

\begin{proof}
Consider arbitrary $\omega_G\in\mathbb D(\mathcal G)$ and $\omega_H\in\mathbb D(\mathcal H)$ of rank one. 
Applying Theorem \ref{thm1} to 
$X(\cdot):=\omega_H\otimes T(\operatorname{tr}_{\mathcal G}(\cdot))\in Q_S(\mathcal H\otimes\mathcal G)$ 
which is obviously a composition of Schr\"odinger quantum channels 
yields a separable Hilbert space $\mathcal K$, a pure $\omega_K\in\mathbb D(\mathcal K)$ and a unitary 
$U\in\mathcal B(\mathcal H\otimes\mathcal G\otimes\mathcal K)$ such that $X$ is of form
\eqref{eq:T_stinespring}. For any $A\in\mathcal B^1(\mathcal H)$ one gets
\begin{align*}
T(A)=\operatorname{tr}_{\mathcal H}(X(A\otimes\omega_G))=(\operatorname{tr}_{\mathcal H}\circ\operatorname{tr}_{\mathcal K})(U(A\otimes \omega_G\otimes\omega_K)U^\dagger)
\end{align*}
with $\omega_G\otimes\omega_K\in\mathbb D(\mathcal G\otimes\mathcal K)$ rank one.
\end{proof}

The last result of this subsection provides a characterization of invertible quantum channels which
leads to a nice simplification later on (cf.~Remark \ref{ch_4_bem_3}). For finite dimensions, this was 
essentially shown in \cite[Coro.~3]{wolf08_1}.

\begin{prop}\label{ch_3_Theorem_14}
Let $T\in Q_S(\mathcal H)$ be bijective. Then the following statements are equivalent.
\begin{itemize}
\item[(a)] $T^{-1}$ is positive.
\item[(b)] There exists unitary $U\in\mathcal B(\mathcal H)$ such that $T(A)=U AU^\dagger$ for all
$A\in\mathcal B^1(\mathcal H)$.
\end{itemize}
In particular if one (and thus both) conditions are fulfilled, then $T\in Q_S(\mathcal H)$ is invertible
as a channel, i.e.~$T$ is bijective and $T^{-1}\in Q_S(\mathcal H)$.
\end{prop}
\begin{proof}
(b)$\,\Rightarrow\,$(a): $\checkmark$ (a)$\,\Rightarrow\,$(b): The proof idea is the same as in 
\cite[Prop.~4.31]{Heinosaari}. Consider the restricted channel 
$T|_{\mathbb D}:\mathbb D(\mathcal H)\to\mathcal B^1(\mathcal H)$ which by assumption is 
convex-linear and injective. As $T$ and $T^{-1}$ are linear, trace-preserving and,
by assumption, positive, the restricted channel satisfies
$$
T(\mathbb D(\mathcal H))\subseteq \mathbb D(\mathcal H) 
\qquad 
T^{-1}(\mathbb D(\mathcal H))\subseteq \mathbb D(\mathcal H)\,,
$$
so $T|_{\mathbb D}:\mathbb D(\mathcal H)\to \mathbb D(\mathcal H)$ is surjective and thus a state
automorphism, i.e.~convex-linear and bijective. 
Then Corollary 3.2 in \cite{Davies} or, more explicitely, Theorem 2.63 in \cite{Heinosaari} 
imply the existence of unitary or anti-unitary $U$ 
such that $T|_{\mathbb D}(\cdot)=U(\cdot)U^\dagger$. If $U$ were anti-unitary, then $T$ would not be
completely positive\cite[Prop.~4.14]{Heinosaari} hence $U$ has to be unitary. Due to 
$\operatorname{span}_{\mathbb C}(\mathbb D(\mathcal H))=\mathcal B^1(\mathcal H)$, this representation 
extends linearily to all of $\mathcal B^1(\mathcal H)$ which concludes the proof.
\end{proof}

\subsection{Dual Channels}\label{sec:dualchannel}
It is well known\cite[Prop.16.26]{MeiseVogt} that the dual space of $\mathcal B^1(\mathcal H)$ is 
isometrically isomorphic to $\mathcal B(\mathcal H)$ by means of the map
$\psi_{\mathcal H}:\mathcal B(\mathcal H)\to (\mathcal B^1(\mathcal H))'$, $B\mapsto\psi_{\mathcal H}(B)$ with
\begin{align*}
(\psi_{\mathcal H}(B))(A):=\operatorname{tr}(BA)
\end{align*}
for all $A\in\mathcal B^1(\mathcal H)$. Note that the weak-$*$-topology and the ultraweak topology on
$\mathcal B(\mathcal H)$ coincide under the above identification
$(\mathcal B^1(\mathcal H))' \cong \mathcal B(\mathcal H)$, cf.~\cite[Section 1.6]{Davies}.

Now, since every positive linear map $T:\mathcal B^1(\mathcal H)\to\mathcal B^1(\mathcal G)$ is bounded 
(cf.~Remark \ref{rem:boundedness}.1) the dual map
\begin{align*}
T':(\mathcal B^1(\mathcal G))'\to(\mathcal B^1(\mathcal H))'\qquad X\mapsto T'(X):=X\circ T
\end{align*}
is well defined and this allows us to construct the so called \emph{dual channel} of $T$
\begin{align*}
T^*:\mathcal B(\mathcal G)\to\mathcal B(\mathcal H)\qquad B\mapsto T^*(B):=(\psi_{\mathcal H}^{-1}\circ T' \circ\psi_{\mathcal G})(B)
\end{align*}
which then satisfies
\begin{align}\label{eq:dual_channel}
\operatorname{tr}(BT(A))=\operatorname{tr}(T^*(B)A)
\end{align}
for all $B\in\mathcal B(\mathcal G)$ and $A\in\mathcal B^1(\mathcal H)$. Alternatively, onc can use 
(\ref{eq:dual_channel}) as defining equation for $T^*$. Furthermore, one has $\Vert T\Vert=\Vert T^*\Vert$
by definition of $T^*$, because $T$ and $T'$ have the same operator norm 
and $\psi_{\mathcal G}$ and $\psi_{\mathcal H}$ are isometric isomorphisms. Some basic properties of $T^*$ are:
\begin{itemize}
\item[(a)] $T^*$ is positive and ultraweakly continuous.
\item[(b)] $T^*$ is completely positive if and only if $T$ is completely positive.
\item[(c)] $T^*$ is unital (i.e.~$T^*(\operatorname{id}_{\mathcal G})=\operatorname{id}_{\mathcal H}$) if and only if $T$ is trace-preserving.
\end{itemize}
For more details and proofs we refer to \cite[p. 35]{Kraus} or \cite[Ch.4.1.2]{Heinosaari}.

\begin{defi}\label{defi_heisenberg_qc}
A Heisenberg quantum channel is a linear, ultraweakly continuous\cite{Note2a}, completely 
positive and unital map $S:\mathcal B(\mathcal G)\to\mathcal B(\mathcal H)$. Furthermore, we define
\begin{align*}
Q_H(\mathcal G,\mathcal H):=\lbrace S:\mathcal B(\mathcal G)\to\mathcal B(\mathcal H)\, |\, S\text{ is Heisenberg quantum channel}\,\rbrace
\end{align*}
and $Q_H (\mathcal H):= Q_H (\mathcal H,\mathcal H)$.
\end{defi}

By the properties listed above, it is evident that the map 
$*:Q_S(\mathcal H,\mathcal G)\to Q_H(\mathcal G,\mathcal H)$
which to any quantum channel assigns its dual channel is well-defined. Furthermore, it is---as we will 
see next---bijective.

\begin{theorem}\label{thm_dual}
\begin{itemize}
\item[(a)] For every $S:\mathcal B(\mathcal G)\to\mathcal B(\mathcal H)$ linear, ultraweakly continuous and positive there exists unique $T:\mathcal B^1(\mathcal H)\to\mathcal B^1(\mathcal G)$ linear and positive such that $T^*=S$.
\item[(b)] For every $S\in Q_H(\mathcal G,\mathcal H)$ there exists unique $T\in Q_S(\mathcal H,\mathcal G)$ such that $T^*=S$.
\end{itemize}
\end{theorem}

\begin{proof}
(a) By the above construction of the dual channel it is obvious that the map $*$ is one-to-one. Therefore,
it suffices to show its surjectivity.

First one shows, similar to  \cite[Ch.~1, Lemma~6.1]{Davies}, that for every positive, linear and ultraweakly continuous functional 
$\lambda:\mathcal B(\mathcal G)\to\mathbb C$ 
there exists a unique positive semi-definite $\rho\in\mathcal B^1(\mathcal G)$ such that 
$\lambda(\cdot)=\operatorname{tr}(\rho(\cdot))$. Next choose arbitrary positive semi-definite $A \in \mathcal B^1(\mathcal H)$ and consider the linear functional 
$$
B \mapsto \operatorname{tr}(S(B)A) 
$$
which by assumption on $S$ is ultraweakly continuous. Our preliminary consideration yields a unique positive semi-definite
$\rho_A\in\mathcal B^1(\mathcal G)$ such that $\operatorname{tr}(S(B)A)=\operatorname{tr}(B\rho_A)$
for all $B\in\mathcal B(\mathcal G)$. This allows to define an $\mathbb R^+$-linear map $\hat T$ on the
positive semi-definite elements of $\mathcal B^1(\mathcal H)$ via $\hat T(A):=\rho_A$. Finally, $\hat T(A)$
can be uniquely extended to a positive, linear map $T:\mathcal B^1(\mathcal H)\to\mathcal B^1(\mathcal G)$
satisfying $T^*=S$.

Now (b) follows from (a) together with the above connections between properties of a positive, linear map and its dual channel.
\end{proof}

\begin{remark}
\begin{enumerate}
\item
Note that, again, boundedness is not required in the definition of a Heisenberg quantum 
channel because, similar to Schr\"odunger quantum channels, they are automatically bounded, 
see Proposition \ref{thm_q_norm_1} below.
\item
In finite dimensions, ultraweak continuity is of course always satisfied and the $*$-map is an 
involution as the sets of trace-class operators and bounded operators coincide.
\end{enumerate}
\end{remark}

\begin{prop}\label{thm_q_norm_1}
Let $T\in Q_S(\mathcal H,\mathcal G)$ and $S\in Q_H(\mathcal G,\mathcal H)$. Then $\|T\|=1$ and $\|S\|=1$.
\end{prop}
\begin{proof}
As each $S\in Q_H(\mathcal G,\mathcal H) $ in particular is linear, positive and unital it has operator norm  
$\|S\|=1$ as a consequence of the Russo-Dye Theorem, cf.~\cite[Cor.~1]{russo1966} or 
Rem.~\ref{rem_19}.1. This directly implies  $\|T\|=\|T^*\|=1$.
\end{proof}

\noindent Alternatively, one can prove Proposition \ref{thm_q_norm_1} via the general Stinespring dilation 
(Corollary \ref{coro_gen_stinespring}) because all maps involved in the Stinespring representation have 
operator norm one. Either way, with this one readily verifies that $Q_H(\mathcal H)$ forms a convex subsemigroup
of the Banach space $\mathcal B(\mathcal B(\mathcal H))$ with unity element 
$\operatorname{id}_{\mathcal B(\mathcal H)}$.\medskip

The partial trace $\operatorname{tr}_{\omega}:\mathcal B(\mathcal H\otimes\mathcal K)\to\mathcal B(\mathcal H)$ 
\textit{with respect to a state} $\omega\in\mathbb D(\mathcal K)$ is defined via 
\begin{align}\label{eq:partial_trace_2}
\operatorname{tr}(\operatorname{tr}_{\omega}(B)A)=\operatorname{tr}(B(A\otimes\omega))
\end{align}
for all $B\in\mathcal B(\mathcal H\otimes\mathcal K)$, $A\in\mathcal B^1(\mathcal H)$, 
cf.~\cite[Ch.~9, Lemma~1.1]{Davies}. Be aware that the map $\operatorname{tr}_{\mathcal K}$ from 
(\ref{eq:partial_trace_1}) and the extension 
$$
i_\omega:\mathcal B^1(\mathcal H)\to \mathcal B^1(\mathcal H\otimes\mathcal K)\qquad A\mapsto A\otimes\omega
$$
with some state 
$\omega\in\mathbb D(\mathcal K)$ are Schr\"odinger quantum channels 
so we immediatly get their dual channels $i_{\omega}^*=\operatorname{tr}_{\omega}$ and 
$\operatorname{tr}_{\mathcal K}^*=i_{\mathcal K}$ with
\begin{align*}
i_{\mathcal K}:\mathcal B(\mathcal H)\to\mathcal B(\mathcal H\otimes\mathcal K)
\qquad B\mapsto B\otimes\operatorname{id}_{\mathcal K}\,.
\end{align*}

\noindent
This leads to the following result.
\begin{coro}\label{coro_1}
For every $S\in Q_H(\mathcal G,\mathcal H)$ there exists a separable Hilbert space $\mathcal K$, 
pure states $\omega_G\in\mathbb D(\mathcal G)$ and $\omega_K\in\mathbb D(\mathcal K)$ and a unitary 
$U\in\mathcal B(\mathcal H\otimes\mathcal G\otimes\mathcal K)$ such that
\begin{align*}
S(B)=(\operatorname{tr}_{\omega_G}\circ\operatorname{tr}_{\omega_K})(U^\dagger(\operatorname{id}_{\mathcal H}\otimes B\otimes \operatorname{id}_{\mathcal K})U)
\end{align*}
for all $B\in\mathcal B(\mathcal G)$. For $\mathcal G=\mathcal H$ this reduces to
\begin{align}\label{eq:stinespring_q_h}
S(B)=\operatorname{tr}_{\omega_K}(U^\dagger(B\otimes \operatorname{id}_{\mathcal K})U)
\end{align}
for all $B\in\mathcal B(\mathcal H)$ where the unitary operator $U$ now acts on $\mathcal H\otimes\mathcal K$.
\end{coro}
\begin{proof}
Note that (\ref{eq:dual_channel}) implies $(T_1\circ T_2)^*=T_2^*\circ T_1^*$ for arbitrary positive, linear 
maps $T_1$ and $T_2$. Hence this is a simple consequence of Theorem \ref{thm1}, Theorem \ref{thm_dual} (b)
and Corollary \ref{coro_gen_stinespring}.
\end{proof}

\begin{remark}
The result in Corollary \ref{coro_1} is a more structured version of Stinespring's theorem \cite{Stinespring}
for Heisenberg quantum channels due to the following: Let $S\in Q_H(\mathcal H)$ (the same argument works
for $S\in Q_H(\mathcal G,\mathcal H)$) and $\omega_K\in\mathbb D(\mathcal K)$ be the state from \eqref{eq:stinespring_q_h} of rank one, 
i.e.~$\omega_K=\langle y,\cdot\rangle y$ for some $y\in\mathcal K$ with $\|y\|=1$. Defining the isometric
embedding $V_y:\mathcal H\to\mathcal H\otimes\mathcal K$, $x\mapsto x\otimes y$, one readily verifies 
via \eqref{eq:partial_trace_2} that $\operatorname{tr}_{\omega_K}(B)=V_y^\dagger BV_y$ for all $B\in\mathcal B(\mathcal H\otimes\mathcal K)$. Now \eqref{eq:stinespring_q_h} becomes
$$
S(\cdot)=V^\dagger\pi(\cdot) V
$$
with the auxiliary Hilbert space $\mathcal H\otimes\mathcal K$ being of tensor form, the 
Stinespring isometry $V = UV_y$ and the unital $*$-homomorphism 
$\pi:\mathcal B(\mathcal H)\to\mathcal B(\mathcal H\otimes\mathcal K) $ being 
$\pi(B):= B\otimes \operatorname{id}_{\mathcal K}$.
To the best of our knowledge, the above representation \eqref{eq:stinespring_q_h} so far only 
appeared in an unpublished (as of now) book by S.~Attal\cite[Thm.~6.15]{Attal6}.
\end{remark}

The above concept of dual channels will be useful to transfer dilation results from the Schr\"odinger
to the Heisenberg picture and vice versa so one is independent of the used quantum-mechanical framework.

\section{Main Results}\label{sec:main_results}

\subsection{Unitary Dilation of Discrete-Time Quantum-Dynamical Systems}\label{sect:unit_dil_A}

\noindent Consider a discrete-time quantum-dynamical system, the evolution of which is described by
\begin{align}
\label{eq:DQDS}
\rho_{n+1}= T(\rho_n), \quad \rho_0\in \mathbb D(\mathcal H)
\end{align}
for arbitrary but fix $T\in Q_S(\mathcal H)$. Obviously, the explicit solution of \eqref{eq:DQDS} is given by
\begin{align*}
\rho_{n}= T^n(\rho_0)
\end{align*}
for all $n \in \mathbb N_0$ By Theorem \ref{thm_monoid}, one has $T^n\in Q_S(\mathcal H)$ and thus Theorem 
\ref{thm1} yields separable Hilbert spaces $\mathcal K_n$, pure states $\omega_n \in\mathbb D(\mathcal K_n)$
and unitaries $U_n\in \mathcal B(\mathcal H\otimes\mathcal K_n )$ such that
\begin{align}\label{eq:3_0}
 T^n(A)=\operatorname{tr}_{\mathcal K_n}\left( U_n(A\otimes\omega_n )U_n^\dagger \right)
\end{align}
for all $A\in \mathcal B^1(\mathcal H)$ and all $n\in\mathbb N_0$. Now our goal is to simplify the
right-hand side of (\ref{eq:3_0}) in the following sense: We want to embed the evolution of $\rho_0$ into 
an evolution of a closed discrete-time quantum-dynamical system, i.e.~we want to replace the r.h.s.~of (\ref{eq:3_0}) by
\begin{align*}
\operatorname{tr}_{\tilde{\mathcal K}}\left( V^n(A\otimes\tilde\omega)(V^\dagger)^n \right)
\end{align*}
where $V$ is an appropriate unitary operator and the separable Hilbert space $\tilde{\mathcal K}$ as well
as the pure state $\tilde\omega$ does no longer depend on $n \in \mathbb N_0$. Our established result reads 
as follows

\begin{theorem}\label{ch_4_satz_1}
For every $ T\in Q_S(\mathcal H)$ there exists a separable Hilbert space $\mathcal K$, a pure state 
$\omega\in\mathbb D(\mathcal K)$ and a unitary $V\in\mathcal B(\mathcal H\otimes\mathcal K)$ such that 
$(\mathcal H\otimes\mathcal{K},(V^n)_{n \in \mathbb Z}, i_\omega, \operatorname{tr}_{\mathcal K})$ is a unitary
dilation of $(T^n)_{n\in\mathbb N_0}$ (in the sense of Definition \ref{def:dilation-Schroedinger-channles}.2). 
In particular, for all $A\in\mathcal B^1(\mathcal H)$ and $n\in\mathbb N_0$, one has
\begin{align}\label{eq:3-00}
 T^n(A)=\operatorname{tr}_{\mathcal K}\left( V^n(A\otimes\omega)(V^\dagger)^n \right)\,.
\end{align}
\end{theorem}


\begin{proof}
First we consider the $n$-dependence of ${\mathcal K}_n$ and $\omega_n$. By construction, cf.~Theorem \ref{thm1}, 
${\mathcal K}_n$ does not depend on $T^n$ anymore, thus we can choose $\tilde{\mathcal K}$ with a countably infinite
basis, for example $\tilde{\mathcal K}=\ell_2(\mathbb N)$, and replace every ${\mathcal K}_n$ with $\tilde{\mathcal K}$. 
Moreover, also by construction, the pure state $\omega_n$ is determined via ${\mathcal K}_n$ and thus can be chosen
independently of $n$, too. Hence we obtain a joint Hilbert space $\tilde{\mathcal K}$ and a pure state $\tilde\omega$
such that
\begin{align*}
 T^n( A)=\operatorname{tr}_{\tilde{\mathcal K}}\left( U_n( A\otimes\tilde\omega )U_n^\dagger \right).
\end{align*}
for all $n\in\mathbb N_0$. Finally, in order to remove the $n$-dependence of the unitary operators $U_n$ we
define $\mathcal K:=\tilde{\mathcal K}\otimes\ell_2(\mathbb Z)$ and $U_n=\operatorname{id}_{\mathcal H\otimes\tilde{\mathcal K} }$
for all $n\leq 0$. Furthermore, let $( e_n)_{n\in\mathbb Z}$ denote the standard basis of $\ell_2(\mathbb Z)$ so 
$\sigma: \ell_2(\mathbb Z) \to \ell_2(\mathbb Z)$ given by $\sigma=\sum_{i\in\mathbb Z}e_ie_{i-1}^\dagger$ yields the right
shift on $\ell_2(\mathbb Z)$. With this 
$U,W:\mathcal B(\mathcal H\otimes\mathcal K)\to\mathcal B(\mathcal H\otimes\mathcal K)$ are defined by
\begin{align*}
{U}:=\sum\nolimits_{n\in\mathbb Z} U_nU_{n-1}^\dagger\otimes e_ne_n^\dagger
\qquad\text{and}\qquad
{W}:= \operatorname{id}_{\mathcal H\otimes\tilde{\mathcal K} }\otimes\,\sigma\,.
\end{align*}
Thus $U$ can be visualised as follows:
\begin{align*}
\underset{\hspace*{-65pt}\begin{matrix}\uparrow & &\end{matrix}}
{\begin{pmatrix}
 \ddots &&&&& \\&\operatorname{id}_{\mathcal H\otimes\tilde{\mathcal K} }&&&& \\&&U_1&&& \\&&&U_2U_1^\dagger&& \\&&&&U_3U_2^\dagger& \\&&&&&\ddots
\end{pmatrix}}
\begin{matrix}
\\\leftarrow\\\\\\\\\\
\end{matrix}\,,
\end{align*}
where the arrows indicate the zero-zero entry of this both-sided ``infinite matrix''. A simple calculation 
shows that $U$, $W$ and therefore also $V:=UW$ are unitary. Next, using the results from Section \ref{sec:dualchannel}, one readily verifies that the maps 
$E:=\operatorname{tr}_{\mathcal K}$ and $J:=i_\omega$ (where 
$\omega:=\tilde\omega \otimes e_0e_0^\dagger\in\mathbb D(\mathcal K)$ is obviously pure) 
satisfy the conditions from Definition \ref{def:dilation-Schroedinger-channles}.1. Then, by induction,
one shows
\begin{equation*}
V^n (A\otimes\omega) (V^\dagger)^n=U_n( A\otimes\tilde\omega ) U_n^\dagger\otimes e_ne_n^\dagger
\end{equation*}
for all $A \in \mathcal B^1(\mathcal H)$ and $n\in\mathbb N_0$. Finally,
$\operatorname{id}_{\tilde{\mathcal K}\otimes\ell_2}=\operatorname{id}_{\tilde{\mathcal K}}\otimes\operatorname{id}_{\ell_2}$ 
implies $\operatorname{tr}_{\tilde{\mathcal K}\otimes\ell_2}=\operatorname{tr}_{\tilde{\mathcal K}}\circ\operatorname{tr}_{\ell_2}$
so
\begin{align*}
\operatorname{tr}_{\mathcal K}\left( V^n(A\otimes\omega)(V^\dagger)^n \right)
= \operatorname{tr}_{\tilde{\mathcal K}}(\operatorname{tr}_{\ell_2}(U_n( A\otimes\tilde\omega ) U_n^\dagger\otimes e_ne_n^\dagger))
= \operatorname{tr}_{\tilde{\mathcal K}}( U_n( A\otimes\tilde\omega )U_n^\dagger ) = T^n( A)
\end{align*}
for all $A \in \mathcal B^1(\mathcal H)$ and $n\in\mathbb N_0$. Hence we constructed a unitary dilation of 
$(T^n)_{n\in\mathbb N_0}$ of the form \eqref{eq:3-00} which concludes the proof.
\end{proof}

\begin{remark}
Note that $i_\omega$ is trace-preserving because $\omega\in\mathbb D(\mathcal K)$, so the above (tensor type) dilation
is trace-preserving.
\end{remark}

Now we can easily extend this result to Heisenberg quantum channels.

\begin{coro}\label{coro_Q_H_1}
For every $S\in Q_H (\mathcal H)$ there exists a separable Hilbert space $\mathcal K$, a pure state 
$\omega\in\mathbb D(\mathcal K)$ and a unitary $V\in\mathcal B(\mathcal H\otimes\mathcal K)$ such that 
$(\mathcal H\otimes\mathcal{K},((V^\dagger)^n)_{n \in \mathbb Z}, i_{\mathcal K}, \operatorname{tr}_{\omega})$ is a 
unitary dilation of $(S^n)_{n\in\mathbb N_0}$ (in the sense of Definition \ref{def:dilation-Heisenberg-channles}.2). 
In particular, for all $B\in\mathcal B(\mathcal H)$ and $n\in\mathbb N_0$, one has
\begin{align*}
  S^n(B)= \operatorname{tr}_\omega\left( (V^\dagger)^n(B\otimes\operatorname{id}_{\mathcal K})V^n \right).
\end{align*}
\end{coro}
\begin{proof}
By Theorem \ref{thm_dual} (b) there exists a unique $T\in Q_S(\mathcal H)$ such that $S=T^*$ and therefore 
$S^n=(T^*)^n=(T^n)^*$. Now Theorem \ref{ch_4_satz_1} yields a separable Hilbert space $\mathcal K$, a pure
state $\omega\in\mathbb D(\mathcal K)$ and a unitary $V$ such that (\ref{eq:3-00}) holds. 
By duality we obtain
\begin{align*}
S^n(B) = \Big(\operatorname{tr}_{\mathcal K}\left( V^n i_\omega(\cdot) (V^\dagger)^n \right)\Big)^*(B)
=\operatorname{tr}_\omega\left( (V^\dagger)^n(B\otimes\operatorname{id}_{\mathcal K})V^n \right)
\end{align*}
for all $B \in \mathcal B(\mathcal H)$ and for all $n\in\mathbb N_0$.
\end{proof}

\begin{remark}\label{rem_subspace}
\begin{enumerate}
\item
Due to $\operatorname{id}_{\mathcal H}\otimes \operatorname{id}_{\mathcal K} = \operatorname{id}_{\mathcal H\otimes\mathcal K}$
we even constructed a unital dilation (of tensor type).
\item
Recall that a ``classical'' unitary dilation $T^n = P_{\mathcal H} \circ U^n \circ \operatorname{inc}_{\mathcal H}$ of 
some Hilbert space contraction $T: \mathcal H \to \mathcal H$ (cf.~Rem.~\ref{rem:classical-dilations}.3), 
where $P_{\mathcal H}$ denotes the orthogonal projection onto $\mathcal H$ and $\operatorname{inc}_{\mathcal H}$
the inclusion map, is called 
\emph{minimal} if the domain of $U\in\mathcal B( \mathcal K )$ is minimal in the sense of
\begin{equation}\label{eq1:classical_dilation}
\mathcal K = \bigvee_{n \in \mathbb Z} U^n \mathcal H\,.
\end{equation}
Here the right-hand side of \eqref{eq1:classical_dilation} denotes the smallest closed subspace of $\mathcal K$
which contains all images $U^n \mathcal H$, $n \in \mathbb Z$, cf.~\cite{NagyFoias1970,FoiasFrazho1990}.
K\"ummerer\cite{Kuemmerer83} captures this idea and defines a dilation $T^n(A) = E\big(\hat{T}^n(J(A))\big)$ of 
an ultraweakly continuous, completely positive and unital map $T:\mathcal A \to \mathcal A$ on a 
$W^*$-algebra $\mathcal A$ to be minimal if 
\begin{equation}\label{eq2:classical_dilation}
\mathcal A = \bigvee_{n \in \mathbb Z}  \hat{T}^n (i(\mathcal A)) 
\end{equation}
holds, where the right-hand side of \eqref{eq2:classical_dilation} now denotes the smallest
closed $W^*$-algebra which contains all images $\hat{T}^n (i(\mathcal A))$, $n \in \mathbb Z$, 
cf.~\cite[Def.~2.1.5]{Kuemmerer83}. It is easy to see that our constructions in Theorem \ref{ch_4_satz_1}  
/ Corollary \ref{coro_Q_H_1} do in general not lead to a minimal dilation in the above sense. However,
one can always restrict a given dilation to the right-hand side of \eqref{eq2:classical_dilation} to obtain 
a minimal one. 
\item
As seen above in \eqref{eq1:classical_dilation} the space 
$\mathcal H^{\infty}_{-\infty} := \bigvee_{n \in \mathbb Z} U^{-n} \mathcal H$ and its forward and backward
invariant counterparts
\begin{equation*}
\mathcal H^{\infty} := \bigvee_{n \in \mathbb N_0} U^n \mathcal H\,,
\quad\text{and}\quad
\mathcal H_{-\infty} := \bigvee_{n \in \mathbb N_0} U^{-n} \mathcal H\,, 
\end{equation*}
play an essential role in the theory of ``classical'' unitary dilations. In particular, they admit 
orthogonal decompositions 
\begin{equation}\label{eq4:classical_dilation}
\mathcal H^{\infty} = \mathcal H \oplus \hat{\mathcal H}^{\infty}\,,
\quad \mathcal H_{\infty} = \mathcal H \oplus \hat{\mathcal H}_{-\infty} 
\quad\text{and}\quad
\mathcal H^{\infty}_{-\infty}  = \hat{\mathcal H}^{\infty} \oplus \mathcal H \oplus \hat{\mathcal H}_{-\infty}
\end{equation}
such that $\hat{\mathcal H}^{\infty}$ and $\hat{\mathcal H}_{-\infty}$ are invariant under $U$ and $U^{-1}$,
respectively, cf.~\cite[Lemma VI.3.1]{FoiasFrazho1990} and \cite{Note3}.
Eventually, \eqref{eq4:classical_dilation} establishes the relation to K\"ummerer's notion of Markovianity, 
cf.~\cite[Prop.~2.2.3 (b)]{Kuemmerer83}.
\end{enumerate}
\end{remark}

Next we want to improve Theorem \ref{ch_4_satz_1} for cyclic $T$, i.e.~in the case of $T^m= T$ for some 
$m\in\mathbb N\setminus\lbrace 1\rbrace$. 
\begin{defi}\label{ch_4_def_1}
In doing so, we define a modified modulo function
\begin{align*}
\nu:\mathbb N\setminus\lbrace 1\rbrace\times\mathbb N&\to\mathbb N\\
(m,n)&\mapsto (n-1)\operatorname{mod}(m-1)+1
\end{align*}
as well as
\begin{align*}
\mu:\mathbb N\setminus\lbrace 1\rbrace\times\mathbb N&\to\mathbb N_0\\
(m,n)&\mapsto \frac{n-\nu(m,n)}{m-1}\,.
\end{align*}
\end{defi}

\noindent
To connect $\nu(m,n)$ to the above cyclicity condition of $T$ we represent $n-1$ as
\begin{align}\label{eq:ch_4_bem_1_0}
n-1=j(m-1)+r
\end{align}
with unique $j\in\mathbb N_0$ and $r\in\lbrace 0,\ldots,m-2\rbrace$. This yields $\nu(m,n)=r+1$ as well as 
$\mu(m,n)=j\in\mathbb N_0$ and we obtain the following result.

\begin{lemma}\label{ch_4_lemma_2}
Let $ T\in Q_S(\mathcal H)$ be cyclic so $ T^m= T$ for $m\in\mathbb N\setminus\lbrace 1\rbrace$. Then
\begin{align*}
 T^n= T^{\nu(m,n)}
\end{align*}
for all $n\in\mathbb N$.
\end{lemma}

\begin{proof}
Via (\ref{eq:ch_4_bem_1_0}) we get $T^n=T^{j(m-1)+r+1}=T^{r+1-j}(T^m)^j=T^{r+1-j}T^j=T^{r+1}=T^{\nu(m,n)}$.
\end{proof}

\noindent
Thus $\mu(m,n)$ indicates how often the cyclicity condition of $T$ can be applied to reduce the exponent $n$ to 
its remaining non-cyclic portion $\nu(m,n)$. With this we obtain the following simplification of Theorem \ref{ch_4_satz_1}.

\begin{theorem}\label{ch_4_koro_1}
Let $ T\in Q_S(\mathcal H)$ be cyclic, i.e.~$ T^m= T$ for some $m\in\mathbb N\setminus\lbrace 1\rbrace$. Then
for the unitary dilation $(\mathcal H\otimes\mathcal{K},(V^n)_{n \in \mathbb Z}, i_\omega, \operatorname{tr}_{\mathcal K})$ 
of $(T^n)_{n\in\mathbb N_0}$ from Theorem \ref{ch_4_satz_1}, one can choose $\mathcal K=\tilde{\mathcal K}\otimes\mathbb C^m$ 
such that (after modifying $V$  and $\omega$ accordingly)
\begin{align*}
 T^n( A)=\operatorname{tr}_{\mathcal K}\left( V^{n+\mu(m,n)}(A\otimes\omega)(V^\dagger)^{n+\mu(m,n)} \right)
\end{align*}
for all $A\in\mathcal B^1(\mathcal H)$ and all $n\in\mathbb N_0$. Note that $\omega\in\mathbb D(\mathcal K)$ still
is a pure state.
\end{theorem}

\begin{proof}
Choose $\tilde{\mathcal K}$ and $\tilde\omega\in\mathbb D(\tilde{\mathcal K})$ as in the proof of Theorem
\ref{ch_4_satz_1}. For every $ T,\ldots, T^{m-1}$ there again exist unitary 
$U_1,\ldots,U_{m-1} \in \mathcal B(\mathcal H\otimes\tilde{\mathcal K})$ satisfying Theorem \ref{thm1}. This
allows to define
\begin{align*}
{U}:=\sum\nolimits_{i=1}^m U_iU_{i-1}^\dagger\otimes e_ie_i^\dagger
\quad\text{and}\quad 
{W}:=\operatorname{id}_{\mathcal H\otimes\tilde{\mathcal K}}\otimes \sum\nolimits_{i=1}^{m} e_{i+1}e_{i}^\dagger
\end{align*}
where $e_{m+1}:=e_1$ and $U_0:=\operatorname{id}_{\mathcal H\otimes\tilde{\mathcal K}}=:U_m$. Then ${W}$ represents a
cyclic shift acting on $\mathbb C^m$ and $U$ is of the following form.
\begin{align*}
{U}=\begin{pmatrix} U_1&&&& \\ &U_2U_1^\dagger&&& \\&&\ddots&&\\&&&U_{m-1}U_{m-2}^\dagger&\\&&&&U_{m-1}^\dagger \end{pmatrix}
\end{align*}
Obviously, $U$, ${W}$ and thus $V:=UW$ are unitary. Again choosing $E:=\operatorname{tr}_{\mathcal K}$ and
$J:=i_\omega$ with pure state $\omega:=\tilde\omega\otimes e_me_m^\dagger\in\mathbb D(\mathcal K)$, one readily verifies
via indiction
\begin{align*}
V^{n+\mu(m,n)} (A\otimes\omega)(V^\dagger)^{n+\mu(m,n)}
=U_{\nu(m,n)}( A\otimes\omega ) U_{\nu(m,n)}^\dagger\otimes e_{\nu(m,n)}e_{\nu(m,n)}^\dagger\,.
\end{align*}
for all $ A\in\mathcal B^1(\mathcal H)$ and $n\in\mathbb N$. Together with Lemma \ref{ch_4_lemma_2} one gets
\begin{align*}
\operatorname{tr}_{\mathcal K}\left( V^{n+\mu(m,n)}(A\otimes\omega)(V^\dagger)^{n+\mu(m,n)} \right)
&=\operatorname{tr}_{\tilde{\mathcal K}}\circ\operatorname{tr}_{\mathbb C^m}\left( U_{\nu(m,n)}( A\otimes\tilde\omega ) U_{\nu(m,n)}^\dagger\otimes e_{\nu(m,n)}e_{\nu(m,n)}^\dagger \right)\\
&=\operatorname{tr}_{\tilde{\mathcal K}}\left( U_{\nu(m,n)}( A\otimes\tilde\omega ) U_{\nu(m,n)}^\dagger\right)
= T^{\nu(m,n)}( A)= T^n( A)
\end{align*}
for all $A\in\mathcal B^1(\mathcal H)$ and $n\in\mathbb N$.
\end{proof}

\begin{remark}\label{ch_4_bem_3}
Note that quantum channels which have an inverse channel (or are ``just'' bijective with positive inverse) can be written as a unitary conjugation $\operatorname{Ad}_U$, cf.~Prop.~\ref{ch_3_Theorem_14}. For such channels, Theorem \ref{ch_4_satz_1} is trivially
fulfilled by choosing $\mathcal K=\mathbb C$, $E=J=\operatorname{id}_{\mathcal B^1(\mathcal H)}$ and $V=U$. The same holds
for cyclic quantum channels which are bijective because cyclicity implies
$T^{-1}=T^{m-2}\in Q_S(\mathcal H)$.
\end{remark}

\subsection{Unitary Dilation of Discrete-Time Quantum-Control Systems}\label{section_control}

Here, we investigate discrete-time quantum-mechanical control systems of the form
\begin{align}\label{eq:chap:3_2_1}
\rho_{n+1}= T_n(\rho_n), \quad \rho_0\in \mathbb D(\mathcal H)
\end{align}
where $T_n$, $n\in\mathbb N_0$ is regarded as control input which can be chosen freely from some subset 
$\mathcal C \subset Q_S(\mathcal H)$. We define $\rho(\cdot,(T_n)_{n\in\mathbb N_0},\rho_0)$ to be the unique 
solution of (\ref{eq:chap:3_2_1}) generated by the control sequence $(T_n)_{n\in\mathbb N_0}$ and the initial
value $\rho_0$. In the sequel, we are interested in whether the dynamics of \eqref{eq:chap:3_2_1}
can be embedded in the dynamics of a unitary discrete-time quantum control system of the same form. 

\begin{defi}\label{ch_4_defi_1}
Let $R_N(\rho_0)$ denote the set of all states which can be reached from $\rho_0$ in $N$ time steps
via (\ref{eq:chap:3_2_1}), i.e.
\begin{align*}
R_N(\rho_0):=\lbrace \rho(N,(T_n)_{n\in\mathbb N_0},\rho_0)\,|\,(T_n)_{n\in\mathbb N_0}\text{ arbitrary control sequence}\rbrace\,.
\end{align*}
Moreover, the overall reachable set of $\rho_0$ is defined by
\begin{align*}
R(\rho_0):=\bigcup\nolimits_{N\in\mathbb N_0}R_N(\rho_0).
\end{align*}
\end{defi}

\noindent
For the remaining section, we assume $\mathcal C:=\lbrace T,S\rbrace$ where $T$ and $S$ are commuting 
but otherwise arbitrary quantum channels over $\mathcal H$. Then the following result is a direct consequence 
of the fact that $T$ and $S$ commute.

\begin{lemma}\label{ch_4_lemma_3}
For all $N\in\mathbb N_0$ one has $R_N(\rho_0):=\lbrace T^kS^{N-k}\rho_0\,|\,k = 0,\ldots,N \rbrace$.
\end{lemma}

Based on this we are interested in dilations of quantum channels of the 
form $T^kS^{N-k}$.

\begin{theorem}\label{ch_4_satz_2}
Let $T,S\in Q_S(\mathcal H)$ be commuting. Then there exists a separable Hilbert space $\mathcal K$, a pure 
state $\omega\in\mathbb D(\mathcal K)$ and unitary $U,V\in\mathcal B(\mathcal H\otimes\mathcal K)$ such that 
\begin{align*}
T^kS^{N-k}( A)=\operatorname{tr}_{\mathcal K}\left( {U}^{k}{V}^{{N-k}}(A\otimes\omega)({V}^\dagger)^{{N-k}}({U}^{\dagger})^{k} \right)\,.
\end{align*}
for all $ A\in\mathcal B^1(\mathcal H)$, $N\in\mathbb N_0$ and $k = 0,\ldots,N$.
\end{theorem}

\begin{proof}
For fixed $N\in\mathbb N$ and $k = 0,\ldots,N$, one has ${T}^{k} {S}^{{N-k}}\in Q_S(\mathcal H)$ 
by Theorem \ref{thm_monoid} and thus Theorem \ref{thm1} yields a separable Hilbert space $\mathcal K_{N,k} $, 
a pure state $\omega_{N,k} \in\mathbb D(\mathcal K_{N,k} )$ and unitary 
$U_{N,k}\in \mathcal B(\mathcal H\otimes\mathcal K )$ such that 
\begin{align*}
{T}^{k}  {S}^{{N-k}} (A) = \operatorname{tr}_{\mathcal K_{N,k}}\left( U_{N,k}( A\otimes\omega_{N,k} )U_{N,k}^\dagger \right).
\end{align*}
for all $ A\in \mathcal B^1(\mathcal H)$. The same line of arguments as in the proof of Theorem 
\ref{ch_4_satz_1} show that $\mathcal K_{N,k}$ and $\omega_{N,k}$ can be chosen
independently of $N$ and $k$,
so there exists some mutual auxiliary space $\tilde{\mathcal K}$ as well as a mutual pure state 
$\tilde\omega\in\mathbb D(\tilde{\mathcal K})$ such that
\begin{align}\label{eq:4-5}
{T}^{k}  {S}^{{N-k}} (A) =\operatorname{tr}_{\tilde{\mathcal K}}\left( U_{N,k}( A\otimes\tilde\omega )U_{N,k}^\dagger \right)
\end{align}
for all $ A\in \mathcal B^1(\mathcal H)$, $N\in\mathbb N$ and $k = 0,\ldots,N$. In particular, 
to every $ {T}^{k}{S}^{N-k}$ we can assign some unitary $U_{N,k}\in\mathcal B(\mathcal H\otimes\tilde{\mathcal K})$
such that (\ref{eq:4-5}) holds. Now, choose 
$\mathcal K:=\tilde{\mathcal K}\otimes\ell_2(\mathbb Z)\otimes \ell_2(\mathbb Z)$ and again 
$E:=\operatorname{tr}_{\mathcal K}$ and $J:=i_\omega$ with pure state 
$\omega:=\tilde\omega \otimes e_0e_0^\dagger\otimes e_0e_0^\dagger\in\mathbb D(\mathcal K)$. 
Moreover, by means of the right shift $\sigma$ from the proof of Theorem \ref{ch_4_satz_1} one defines
\begin{align*}
{W}_1:=\operatorname{id}_{\mathcal H\otimes\tilde{\mathcal K}}\otimes \,\sigma\otimes \,\sigma\,,\qquad
{U}_1&:=\sum_{m,n\in\mathbb Z} U_{m,n}U_{m-1,n-1}^\dagger \otimes e_me_m^\dagger\otimes e_ne_n^\dagger\,,\\
{W}_2:=\operatorname{id}_{\mathcal H\otimes\tilde{\mathcal K}}\otimes\,\sigma\otimes\operatorname{id}_{\ell_2}\,,\qquad
{U}_2&:=\sum_{n\in\mathbb Z}U_{n,0}U_{n-1,0}^\dagger \otimes e_ne_n^\dagger \otimes \operatorname{id}_{\ell_2}\,,
\end{align*}
where $U_{m,n}:=\operatorname{id}_{\mathcal H\otimes\tilde{\mathcal K}}$ if $m<1$ or $n\notin\lbrace 0,\ldots,m\rbrace$. 
Obviously, $W_1$ and $W_2$ are unitary. The unitarity of $U_1,$ and $U_2$ is readily verified via
the unitarity of $U_{N,k}$ so $U:=U_1W_1$ and $V:=U_2W_2$ are unitary, too. As before, by induction one shows 
\begin{align}\label{eq:commuting_lemma_1}
{V}^{j}(A\otimes\omega) ({V}^\dagger)^{j}=U_{{j},0} ( A\otimes\tilde\omega) U_{{j},0}^\dagger\otimes e_{j}e_{j}^\dagger\otimes e_0e_0^\dagger
\end{align}
for all $ A\in\mathcal B^1(\mathcal H)$ and $j \in \mathbb N_0$
and based on this
\begin{align}\label{eq:4-7}
 {U}^{k}{V}^{{N-k}}(A\otimes\omega)({V}^\dagger)^{{N-k}}({U}^{\dagger})^{k} =U_{N,k} ( A\otimes\tilde\omega) U_{N,k}^\dagger\otimes e_{N}e_{N}^\dagger\otimes e_ke_k^\dagger
\end{align}
for all $ A\in\mathcal B^1(\mathcal H)$, $N\in\mathbb N_0$ and $k = 0,\ldots,N$.
Note that the case $k=0$ reproduces (\ref{eq:commuting_lemma_1}) and thus can be omitted. 
Finally, \eqref{eq:4-5} and \eqref{eq:4-7} imply
\begin{align*}
\operatorname{tr}_{\mathcal K}\left( {U}^{k}{V}^{{N-k}}(A\otimes\omega)({V}^\dagger)^{{N-k}}({U}^{\dagger})^{k} \right)
&=\operatorname{tr}_{\tilde{\mathcal K}}(\operatorname{tr}_{\ell_2\otimes \ell_2}( U_{N,k} ( A\otimes\tilde\omega) U_{N,k}^\dagger\otimes e_{N}e_{N}^\dagger\otimes e_{k}e_{k}^\dagger ))\\
&=\operatorname{tr}_{\tilde{\mathcal K}}( U_{N,k}( A\otimes\tilde\omega )U_{N,k}^\dagger )= {T}^{k}  {S}^{{N-k}}  (A)
\end{align*}
for all $ A\in\mathcal B^1(\mathcal H)$, $N\in\mathbb N$ and $k = 0,\ldots,N$ which concludes this proof.
\end{proof}

\begin{remark}
The statement of Theorem \ref{ch_4_satz_2} can be extended to finitely many commuting channels 
$ T_1,\ldots, T_m\in Q_S(\mathcal H)$. Obviously, it is natural to choose
\begin{align*}
\mathcal K=\tilde{\mathcal K}\otimes\underbrace{\ell_2(\mathbb Z)\otimes\ldots\otimes \ell_2(\mathbb Z)}_{m\text{-times}}
\end{align*}
as common auxiliary space. The rest of the proof is completely analogous.
\end{remark}

We can now transfer the above result to obtain a characterization of the reachable set of the control 
system (\ref{eq:chap:3_2_1}).

\begin{coro}\label{ch_4_koro_3}
There exists a separable Hilbert space $\mathcal K$, a pure state $\omega\in\mathbb D(\mathcal K)$ and 
unitary $U,V\in\mathcal B(\mathcal H\otimes\mathcal K)$ such that
\begin{align*}
\rho(N,(T_n)_{n\in\mathbb N_0},\rho_0)=\operatorname{tr}_{\mathcal K}\left( {U}^{k}{V}^{{N-k}}(\rho_0\otimes\omega)({V}^\dagger)^{{N-k}}({U}^{\dagger})^{k} \right)
\end{align*}
for all controls $(T_n)_{n\in\mathbb N_0}$, initial states $\rho_0\in\mathbb D(\mathcal H)$ and $N\in\mathbb N_0$,
where $k=k(N,(T_n)_{n \in\mathbb N_0}) \in \lbrace 0,\ldots,N\rbrace$ counts how often $T$ occurs 
in the control sequence $(T_n)_{n\in\mathbb N_0}$ during the first $N$ time steps.	
\end{coro}

\begin{proof}
By Definition \ref{ch_4_defi_1}, $\rho(N,(T_n)_{n\in\mathbb N_0},\rho_0)\in R_N(\rho_0)$ and hence by Lemma 
\ref{ch_4_lemma_3} there exists $k\in\lbrace 0,\ldots,N\rbrace$ such that $\rho(N,(T_n)_{n\in\mathbb N_0},\rho_0)=T^kS^{N-k}(\rho_0)$. Thus the result follows immeditely from Theorem \ref{ch_4_satz_2}.
\end{proof}

\begin{coro}\label{ch_4_koro_2}
Let $\mathcal K$, $\omega\in\mathbb D(\mathcal K)$ and $U,V\in\mathcal B(\mathcal H\otimes\mathcal K)$ 
be as in Corollary \ref{ch_4_koro_3}. Then, for all $N\in\mathbb N_0$ and $\rho_0\in\mathbb D(\mathcal H)$
one has 
\begin{align}\label{eq:ch_4_koro_2_1}
R_N(\rho_0)\subseteq \operatorname{tr}_{\mathcal K}(\tilde R_N(\rho_0\otimes\omega))
\end{align}
and thus $R(\rho_0)\subseteq \operatorname{tr}_{\mathcal K}(\tilde R(\rho_0\otimes\omega))$. Here, 
$\tilde R(\tilde\rho_0)$ and $\tilde R_N(\tilde\rho_0)$ denote the reachable sets of 
the discrete-time closed quantum control system
\begin{align*}
\tilde \rho_{n+1}=U_n\tilde \rho_n U_n^\dagger\,, \quad \tilde\rho_0\in\mathbb D(\mathcal H\otimes\mathcal K)
\end{align*}
with $U_n\in\lbrace U,V\rbrace$ for all $n\in\mathbb N_0$.
\end{coro}

\begin{proof}
By Lemma \ref{ch_4_lemma_3} and Theorem \ref{ch_4_satz_2}, one has
\begin{align*}
R_N(\rho_0)&=\lbrace T^kS^{N-k}(\rho_0)\,|\,k = 0,\ldots,N \rbrace
\\&=\lbrace \operatorname{tr}_{\mathcal K}\left( {U}^{k}{V}^{{N-k}}(\rho_0\otimes\omega)({V}^\dagger)^{{N-k}}({U}^{\dagger})^{k} \right)\,|\,k = 0,\ldots,N\rbrace
\\&=\operatorname{tr}_{\mathcal K}(\lbrace  {U}^{k}{V}^{{N-k}}(\rho_0\otimes\omega)({V}^\dagger)^{{N-k}}({U}^{\dagger})^{k} \,|\,k = 0,\ldots,N \rbrace)
\subseteq \operatorname{tr}_{\mathcal K}(\tilde R_N(\rho_0\otimes\omega))\,.\qedhere
\end{align*}
\end{proof}

\begin{remark}\label{ch_4_bem_7}
\begin{enumerate}
\item
Note that the unitary channels $U$ and $V$ of Corollary \ref{ch_4_koro_2} do in general not commute, 
so (\ref{eq:ch_4_koro_2_1}) states a proper inclusion rather than an equality for $N>1$.
\item
Consider the dual problem of (\ref{eq:chap:3_2_1}), i.e.~let $T,S \in Q_H(\mathcal H)$ be 
two commuting Heisenberg channels. Of course, one can translate the above results---which we will omit 
here---into the Heisenberg picture via Corollary \ref{coro_Q_H_1}.
However, we want to comment on the result of Davies\cite{davies78} which was already mentioned in the 
introduction and yields a unitary dilation with \emph{commuting} unitary channels, at the cost of our desired
partical trace structure.

Let $G=(\mathbb Z\times\mathbb Z,+)$ with subgroup
$\mathbb S:=\lbrace (N,k)\in G\,|\, N\in\mathbb N_0\text{ and } 0\leq k\leq N\rbrace$ and define the family
$(T_g)_{g\in G}$ of Heisenberg channels via $T_g:= T^k S^{N-k}$ for
$g = (N,k)\in \mathbb S$ and $T_g:=\operatorname{id}_{\mathcal B(\mathcal H)}$ otherwise. Adjusting the proof 
of{\,}\cite[Thm.~3.1]{davies78} to discrete groups and using Corollary \ref{coro_1}, one gets a Hilbert space 
$\mathcal K$, a unitary representation $U$ of $G$ on $\mathcal H\otimes\mathcal K$ and a conditional 
expectation $E$ such that $T_g(B)=E(U_g(B\otimes\operatorname{id}_{\mathcal K})U_g^\dagger)$ for all 
$B\in\mathcal B(\mathcal H)$. As $(U_g)_{g\in G}$ is a representation of $G$ we may consider the commuting 
unitary operators $U_{(1,1)}=:U$, $U_{(1,0)}=:V$ resulting in
\begin{equation*}
T^k S^{N-k}(B)= T_g(B) = E({U}^{k}{V}^{{N-k}}(B\otimes\operatorname{id}_{\mathcal K})({V}^\dagger)^{{N-k}}({U}^{\dagger})^{k})
\end{equation*}
for all $B\in\mathcal B(\mathcal H)$, $N\in\mathbb N_0$ and $k = 0,\ldots,N$. Observe that we did not 
use the fact that $T,S$ commute so this result even holds for arbitrary channels $T$ and $S$ with the 
drawback of $E$ lacking any partial trace structure, see also \cite{Note1}.
\end{enumerate}
\end{remark}

\begin{acknowledgments}
This article is based on a master thesis\cite{vomEnde} which was written at the Institute of 
Mathematics of the University of W\"urzburg. The authors are grateful to Michael M.~Wolf for
drawing their attention to several more recent publications on dilations of completely positive maps.
\end{acknowledgments}

\appendix


\section{Topological Properties of $Q_S (\mathcal H)$}\label{app:monoid}

For the following definition, we refer to \cite[Ch.VI.1]{dunford1963linear}.

\begin{defi}\label{topologies}
Let $\mathcal X$ and $\mathcal Y$ be arbitrary Banach spaces.
\begin{itemize}
\item[(a)] 
The strong operator topology (s.o.t.) on $\mathcal B(\mathcal X,\mathcal Y)$ is the locally convex 
topology induced by the family of seminorms of the form $T\to\Vert Tx\Vert$ with $x\in \mathcal X$.
\item[(b)]
The weak operator topology (w.o.t.)~on $\mathcal B(\mathcal X,\mathcal Y)$ is the locally convex topology 
induced by the family of seminorms of the form $T\to|y(Tx)|$ with $(x,y)\in \mathcal X\times \mathcal Y'$.
\end{itemize}
\end{defi}

\noindent
Note that both topologies, the s.o.t.~as well as the w.o.t., are Hausdorff so limits are unique.\medskip 

By the natural isomorphism $(\mathcal B^1(\mathcal H))' \cong \mathcal B(\mathcal H)$, see Section 
\ref{sec:dualchannel}, one has the following equivalence: A net $(T_\alpha)_{\alpha\in I}$ in
$\mathcal B(\mathcal B^1(\mathcal H))$ converges to $T\in\mathcal B(\mathcal B^1(\mathcal H))$ 
in w.o.t.~if and only if
\begin{align}\label{eq:AppA_1}
\lim_{\alpha\in I}|\operatorname{tr}(BT_\alpha(A))-\operatorname{tr}(BT(A))|=0
\end{align}
for all $A\in\mathcal B^1(\mathcal H)$ and $B\in\mathcal B(\mathcal H)$. 

\begin{remark}[Metrizability of s.o.t.~and w.o.t.~on bounded subsets]\label{rem_metrizable}
At this point one might ask whether the strong or weak operator topology is metrizable. If this is the case, closed 
and sequentially closed sets do coincide which, of course, is of interest for further investigations. 
The following is well known in the literature, cf.~\cite[Thm.~1.2 and 1.13]{Kim}: If $\mathcal X$ is 
separable, then the s.o.t.~is metrizable on bounded subsets of $\mathcal B(\mathcal X)$. If $\mathcal X'$
is also separable, then the w.o.t.~is metrizable on bounded subsets of $\mathcal B(\mathcal X)$.

Now, recall that $\mathcal H$ is assumed to be separable. Therefore it is evident that the subspace of 
finite-rank operators $\mathcal F(\mathcal H)$ and hence $\mathcal B^1(\mathcal H)$ itself, which is 
the $\nu_1$-closure of $\mathcal F(\mathcal H)$ (cf.~\cite[Lemma XI.9.11]{dunford1963linear}), is separable.
Moreover, we already know from Proposition \ref{thm_q_norm_1} that $Q_S(\mathcal H)$ is a subset of the unit ball 
in $\mathcal B(\mathcal B^1(\mathcal H))$. This implies that the s.o.t.~on $Q_S(\mathcal H)$ is metrizable 
and thus convergence, closedness, continuity, etc.~can be fully characterized by sequences. On the other 
hand, it is also well known that $\mathcal B(\mathcal H)$ is not separable with respect to the operator norm 
topology as the non-separable space $\ell^\infty$ can be isometrically embedded into $\mathcal B(\mathcal H)$. 
Hence $(\mathcal B^1(\mathcal H))'$ is not separable and the above metrizability result does not apply to 
the w.o.t.~on $Q_S(\mathcal H)$.

However, one could make use of the result that for \emph{convex sets} in $\mathcal B(\mathcal B^1(\mathcal H))$, the closures with respect to the w.o.t.~and 
the s.o.t coincide, cf.~\cite[Coro.~VI.1.6]{dunford1963linear}. Therefore, in the 
proof of Theorem \ref{thm_monoid} one could focus on the s.o.t. On the other hand, Lemma \ref{lemma_2} ff.~show that a direct approach
via the w.o.t.~is just as simple.
\end{remark}

For clarity of the proof of Theorem \ref{thm_monoid}, we first state some auxiliary results.  

\begin{lemma}\label{lemma_2_b}
For every linear map $S:\mathcal B^1(\mathcal H)\to\mathcal B^1(\mathcal G)$ the following statements
are equivalent.
\begin{itemize}
\item[(a)]
$S$ is positive.
\item[(b)]
For all $ A\in\mathcal B^1(\mathcal H)$ and $B\in\mathcal B(\mathcal G)$ with $A,B\geq 0$, one has 
$\operatorname{tr}(BS(A))\geq 0$.
\end{itemize}
\end{lemma}

\begin{proof}
(a)$\,\Rightarrow\,$(b): For $A,B\geq 0$ and $S$ positive, we obtain $S(A) \geq 0$ and thus 
\begin{equation*}
\operatorname{tr}(BS(A)) = \operatorname{tr}(\sqrt B S(A)\sqrt B) \geq 0\,,
\end{equation*}
where $\sqrt B\geq 0$ denotes the unique square root of $B$.

\noindent
(b)$\,\Rightarrow\,$(a): Choosing $B :=\langle x,\cdot\rangle x$ for arbitrary $x\in\mathcal G$
yields $B \geq 0$ and 
\begin{align*}
\langle x,S(A)x\rangle=\operatorname{tr}(BS(A))\geq 0\,,
\end{align*}
for all $A \geq 0$. Hence it follows $S(A)\geq 0$ so $S$ is positive.
\end{proof}


\begin{lemma}\label{lemma_2}
Let $(T_\alpha)_{\alpha\in I}$ be a net in $\mathcal B(\mathcal B^1(\mathcal H))$ which converges to $T\in\mathcal B(\mathcal B^1(\mathcal H))$ in w.o.t. Then the following statements hold.
\begin{itemize}
\item[(a)] If $T_\alpha$ is trace-preserving for all $\alpha\in I$ then $T$ is trace-preserving.
\item[(b)] If $T_\alpha$ is positive for all $\alpha\in I$ then $T$ is positive.
\end{itemize}
\end{lemma}

\begin{proof}
Both statements follow from (\ref{eq:AppA_1}): (a) by choosing $B=\operatorname{id}_{\mathcal H}$ 
and (b) by applying Lemma \ref{lemma_2_b} and taking into account that $[0,\infty)$ is a closed subset
of $\mathbb R$.
\end{proof}

%
For the proof of our next result we recall that $\mathcal B^1(\mathcal H\otimes\mathbb C^m)$ and 
$\mathcal B^1(\mathcal H)\otimes\mathbb C^{m\times m}$ can be identified as follows. Any 
$A\in\mathcal B(\mathcal H\otimes\mathbb C^m)$ can be represented as $A=\sum_{i,j=1}^m A_{ij}\otimes E_{ij}$ 
with the standard basis $(E_{ij})_{i,j=1}^m$ of $\mathbb C^{m\times m}$ and appropriate 
$A_{ij}\in\mathcal B(\mathcal H)$. Then, the following statements are equivalent 
\cite[p.~33-34]{Kraus}.
\begin{itemize}
\item[(a)] $A\in\mathcal B^1(\mathcal H\otimes\mathbb C^m)$
\item[(b)] $A_{ij}\in\mathcal B^1(\mathcal H)$ for all $i,j\in\lbrace 1,\ldots,m\rbrace$
\end{itemize}
%

\begin{lemma}\label{lemma_3}
Let $(T_\alpha)_{\alpha\in I}$ be a net in $\mathcal B(\mathcal B^1(\mathcal H))$ converging to $T\in\mathcal B(\mathcal B^1(\mathcal H))$ in w.o.t.
Then, for all $m\in\mathbb N$, the net 
$(T_\alpha\otimes\operatorname{id}_m)_{\alpha\in I}$
converges to
$T\otimes\operatorname{id}_m\in\mathcal B(\mathcal B^1(\mathcal H\otimes\mathbb C^m))$ in w.o.t.
\end{lemma}

\begin{proof}
According to (\ref{eq:AppA_1}) we have to show
\begin{align}\label{eq:App1_2}
\lim_{\alpha\in I}|\operatorname{tr}(B(T_\alpha\otimes\operatorname{id}_m-\,T\otimes\operatorname{id}_m)A)|= 0
\end{align}
for all $A\in\mathcal B^1(\mathcal H\otimes\mathbb C^m)$ and $B\in\mathcal B(\mathcal H\otimes\mathbb C^m)$. 
As seen above 
every $A\in\mathcal B^{1}(\mathcal H\otimes\mathbb C^m)$ and $B\in\mathcal B(\mathcal H\otimes\mathbb C^m)$ 
can be represented as finite linear combinations of elements 
$A_{ij}\otimes E_{ij}\in \mathcal B^{1}(\mathcal H)\otimes \mathbb C^{m\times m}$
and $B_{ij}\otimes E_{ij}\in\mathcal B(\mathcal H)\otimes \mathbb C^{m\times m}$, respectively,
with $i,j = 1,\ldots,m$. Hence
\begin{align*}
\operatorname{tr}(B(T_\alpha\otimes\operatorname{id}_m-\,T\otimes\operatorname{id}_m)A) = 
\operatorname{tr}(B((T_\alpha-T)\otimes\operatorname{id}_m)A)=\sum\nolimits_{i,j=1}^m \operatorname{tr}(B_{ij}(T_\alpha-T)(A_{ji}))
\end{align*}
so convergence of $(T_\alpha\otimes\operatorname{id}_m)_{\alpha\in I}$ can easily be related to the convergence 
of $(T_\alpha)_{\alpha\in I}$.
\end{proof}

Now for the main proof of this section.

\begin{proof}[Proof of Theorem \ref{thm_monoid}]
Since every linear and positive operator on $\mathcal B^1(\mathcal H)$ is naturally norm bounded 
as a simple consequence of{\,} \cite[Ch.~2, Lemma~2.1]{Davies}
, the
set $Q_S(\mathcal H)$ of all Schr\"odinger channels is a bounded subset of 
$\mathcal B(\mathcal B^1(\mathcal H))$. Now it is readily verified that $Q_S(\mathcal H)$ is a 
convex subsemigroup of $\mathcal B(\mathcal B^1(\mathcal H))$, cf.\cite[Ch.4.3]{Heinosaari}.
Next consider a net $(T_\alpha)_{\alpha\in I}$ in $Q_S(\mathcal H)$ converging to some $T\in \mathcal B(\mathcal B^1(\mathcal H))$ in w.o.t. 
By Lemma \ref{lemma_2} (a), the map $T$ is trace-preserving and by Lemma \ref{lemma_3} 
$(T_\alpha\otimes\operatorname{id}_m)_{\alpha\in I}$ converges to $T\otimes\operatorname{id}_m$ 
with respect to the w.o.t. 
Then applying Lemma \ref{lemma_2} (b) to the net $(T_\alpha\otimes\operatorname{id}_m)_{\alpha\in I}$ yields 
that $T$ is also $m$-positive for all $m\in\mathbb N$. Hence $T$ is a Schr\"odinger quantum channel 
and $Q_S(\mathcal H)$ is closed in $\mathcal B(\mathcal B^1(\mathcal H))$ with respect to the w.o.t. 
The well-known fact that the w.o.t.~is weaker than the s.o.t.~and the uniform operator topology 
concludes the proof.
\end{proof}

\begin{remark}
Note that in the above proof we did not explicitely use the fact that domain and range of the operator 
$T$ coincides. Therefore, the convexity and closedness results trivially extend to $Q_S(\mathcal H,\mathcal G)$.
\end{remark}

\section{Glossary on Dilations}\label{sec:ditalions}

\noindent For the sake of self-containedness, we recall some basic terminology concerning different types
of dilations of linear contractions. Let us start with the Banach space case.

\begin{defi}\label{def:dilation-Banach-space}
Let $\mathcal X$ be an arbitrary Banach space.
\begin{enumerate}
\item
Let $T:\mathcal X \to \mathcal X$ be a linear contraction, i.e.~$\|T\| \leq 1$. A dilation $(\mathcal Y,\hat T,J,E)$ of $T$ consists
of a Banach space $\mathcal Y$ and a triple of maps $(\hat{T}, J, E)$ with
\begin{equation}\label{eq:dilation-general}
T = E \circ \hat{T} \circ J \quad \text{and} \quad E \circ J = \operatorname{id}_{\mathcal X}\,,
\end{equation}
where the linear maps $\hat{T}$, $J$ and $E$ satisfy:
\begin{enumerate}
\item
$\hat{T}: \mathcal Y \to \mathcal Y$ is a bi-isometry (i.e.~$\hat T$ is bijective and $\hat T,\hat T^{-1}$ are isometries) 
\item
$J: \mathcal X \to \mathcal Y$ is an isometric embedding of $\mathcal X$ in $\mathcal Y$.
\item
$E: \mathcal Y \to \mathcal X$ 
has operator norm $\|E\| = 1$.
\end{enumerate}
\item
Let $S \subset G$ be a semigroup of a group $G$ and $(T_g)_{g \in S}$ be a representation of $S$ with 
values in the contraction semigroup of $\mathcal X$. A dilation $(\mathcal Y,(\hat{T}_g)_{g \in G_0}, J, E)$
of $(T_g)_{g \in S}$ consists of a Banach space $\mathcal Y$, a subgroup $G_0 \subset G$ and a triple 
$((\hat{T}_g)_{g \in G_0}, J, E)$ with
\begin{equation*}
T_g = E \circ \hat{T}_g \circ J \quad \text{and} \quad E \circ J = \operatorname{id}_{\mathcal X}
\end{equation*}
for all $g \in S\subset G_0$, where $(\hat{T}_g)_{g \in G_0}$ is a linear representation of $G_0$
with values in the isometry group of $\mathcal Y$ and $J$, $E$ as before.
\end{enumerate}
\end{defi}

\begin{remark}\label{rem:classical-dilations}
\begin{enumerate}
\item
Note that $E \circ J = \operatorname{id}_{\mathcal X}$ implies that $E$ is onto and $J$ is injective. 
Furthermore, $J \circ E:\mathcal Y \to \mathcal Y$ is a projection of norm $1$ from 
$\mathcal Y$ onto the range of $J$.
\item
If $S$ is assumed to be abelian and there exists a ``dilation'' of $(T_g)_{g \in S}$ such that
$\hat{T}_g$ is well-defined for all $g \in S$ then $(\hat{T}_g)_{g \in S}$ obviously extends
to a proper dilation in the above sense, where $G_0$ can be chosen to be the subgroup generated 
by $S$. For non-abelian $S$, however, this extension property is not obvious.
\item 
Choosing $S := \mathbb{N}_0$ and $T_n := T^n$ in Def.~\ref{def:dilation-Banach-space}.2, where 
$T:\mathcal H \to \mathcal H$ is a linear contraction, we recover the ``classical'' concept of a
linear diliation (see also Rem.~\ref{rem_subspace}.2). To distinguish such a dilation of $T$---which in our sense is actually a dilation of $(T^n)_{n\in\mathbb N_0}$---from a dilation of $T$ in the sense of 
Def.~\ref{def:dilation-Banach-space}.1 one sometimes calls the latter a ``dilation of first order'', 
cf.~\cite{Kuemmerer83}.
\item
Once continuity comes into play, things become more sublte as one can either require that the continuity 
properties of $g \mapsto T_g$ are preserved by $g \mapsto \hat{T}_g$ (which in some cases is
unfeasible) or allow that the continuity is relaxed, cf., e.g., \cite[Rem.~17.5]{EvansLewisOverview77}.
For our applications, however, this is not an issue as we are only concerned with the case
$S = \mathbb{N}_0$.
\end{enumerate}
\end{remark}

In the context of quantum channels (or, more generally, completely positive maps) various specializations of
the above definitions to 
\begin{itemize}
\item abstract $C^*$- or $W^*$-algebras
\item Heisenberg quantum channels 
\item and Schr\"odinger quantum channels
\end{itemize}
are available in the literature. For more details we refer to \cite{EvansLewisOverview77,Gaebler2015}.

\begin{defi}\label{def:dilation-von-Neumann-algebras}
Let $\mathcal A$ be a unital $C^*$-algebra.
\begin{enumerate}
\item
Let $T: \mathcal A \to \mathcal A$ be linear, completely positive and unital (i.e.~identity preserving).
A dilation $(\mathfrak A,\hat T,J,E)$ of $T$ consists of a unital $C^*$-algebra $\mathfrak A$
and a triple of maps $(\hat T,J,E)$ with
\begin{equation}\label{eq:cp-dilation}
T = E \circ \hat{T} \circ J \quad \text{and} \quad E \circ J = \operatorname{id}_{\mathcal A}\,,
\end{equation}
where $\hat{T}$, $J$ and $E$ satisfy:
\begin{enumerate}
\item
$\hat T:\mathfrak A\to\mathfrak A$ is a $*$-automorphism. 
\item
$J: \mathcal A \to \mathfrak A$ is a $*$-homomorphism of $\mathcal A$ into $\mathfrak A$.
\item
$E:\mathfrak A\to\mathcal A$ is linear and completely positive 
with operator norm $\|E\| = 1$.
\end{enumerate}
\item
Let $S \subset G$ be a semigroup of a group $G$ and let $(T_g)_{g\in S}$ be a semigroup representation
of $S$ with values in the set of completely positive, unital maps on $\mathcal A$. A dilation
$(\mathfrak A,(\hat{T}_g)_{g \in G_0}, J, E)$ of $(T_g)_{g \in S}$ consists of a unital 
$C^*$-algebra $\mathfrak A$, a subgroup $G_0 \subset G$ and a triple $((\hat{T}_g)_{g \in G_0}, J, E)$ with
\begin{equation*}
T_g = E \circ \hat{T}_g \circ J \quad \text{and} \quad E \circ J = \operatorname{id}_{\mathcal A}\,,
\end{equation*}
for all $g \in S\subset G_0$, where $(\hat{T}_g)_{g \in G_0}$ is a representation of $G_0$ with values in the
$*$-automorphism group of $\mathfrak A$ and $J$, $E$ as before.
\end{enumerate}
If, in addition, $J(\operatorname{id}_{\mathcal A})=\operatorname{id}_{\mathfrak A}$ then the dilation is said to be unital. On the other hand, if $\mathcal A$ is even a $W^*$-algebra, then all involved maps are in general assumed to 
be ultraweakly continuous. 
\end{defi}

\begin{remark}\label{rem_19}
\begin{enumerate}
\item
Let $\mathcal A$, $\mathcal B$ be unital $C^*$-algebras and let $T: \mathcal A \to \mathcal B$ be 
unital. Then positivity of $T$ is equivalent to the norm condition $\|T\| = 1$, cf.~\cite{russo1966,Blackadar06}. 
In particular, one has $\|T\| = \|T(\operatorname{id_{\mathcal A}})\|$ for every positive map 
$T: \mathcal A \to \mathcal B$ so unitality of $T$ implies that $T$ is a contraction.
\item
As every $*$-homomorphism is trivially completely positive and every injective $*$-homomorphism is always 
isometric, \eqref{eq:cp-dilation} yields a dilation in the sense of \eqref{eq:dilation-general}. 
Moreover, if a dilation is unital, then $E$ is unital as well because \eqref{eq:cp-dilation} implies 
$\operatorname{id}_{\mathcal A}=(E\circ J)(\operatorname{id}_{\mathcal A})=E(\operatorname{id}_{\mathfrak A})$.
\item
Every $*$-automorphism $\hat T:\mathfrak A\to \mathfrak A$ on a unital $C^*$-algebra $\mathfrak A$ is unital itself because of
$$
\hat T(\operatorname{id}_{\mathfrak A})=\hat T(\operatorname{id}_{\mathfrak A})\operatorname{id}_{\mathfrak A}=\hat T(\operatorname{id}_{\mathfrak A})\hat T(\hat T^{-1}(\operatorname{id}_{\mathfrak A}))=\hat T(\operatorname{id}_{\mathfrak A} T^{-1}(\operatorname{id}_{\mathfrak A}))=\operatorname{id}_{\mathfrak A}\,.
$$
\item
Let $A$ be a $C^*$-subalgebra of a unital $C^*$-algebra $\mathfrak A$, i.e. 
$\operatorname{id}_{\mathfrak A} \in \mathcal A \subset \mathfrak A$.  Then a linear map $E:\mathfrak A\to\mathcal A$ 
is said to be a \emph{conditional expectation} (of $\mathfrak A$ onto $\mathcal A$) if it is 
completely positive with norm $\|E\| =1$ and satisfies
\begin{equation}\label{eq:cond_exp-1}
E(AB)=AE(B) \quad \text{for all $A \in\mathcal A$ and $B\in\mathfrak A$}\,.
\end{equation} 
Obviously, \eqref{eq:cond_exp-1} implies that $E$ is a unital (cf.~Rem.~\ref{rem_19}.1) projection onto 
$\mathcal A$, that is $E(A) = A$ for all $A\in\mathcal A$. The converse is also true, i.e.~every projection $E:\mathfrak A\to\mathcal A$ of norm 
$\|E\| =1$ is a conditional expectation, cf.~\cite[Thm.~II.6.10.2]{Blackadar06} and \cite{tomiyama1957}. 
Moreover, exploiting that $E(B^*) = E(B)^*$ for all $B\in\mathfrak A$, which results from
the (complete) positivity of $E$, one can easily show that \eqref{eq:cond_exp-1} is equivalent to 
\begin{equation*} 
E(BA) = E(B)A \quad \text{for all $A \in\mathcal A$ and $B\in\mathfrak A$}
\end{equation*} 
and, since $\mathcal A$ is unital, also to
\begin{equation} \label{eq:cond_exp}
E(A_1BA_2)=A_1E(B)A_2 \quad \text{for all $A_1,A_2\in\mathcal A$ and $B\in\mathfrak A$}\,.
\end{equation}
In the literature, \eqref{eq:cond_exp-1} is often replaced by the ``more symmetric'' condition
\eqref{eq:cond_exp}.
Now if $\mathcal A\not\subset\mathfrak A$, but $\mathcal A$ can be embedded into $\mathfrak A$ via some unital, 
injective $*$-homomorphism $J:\mathcal A\to\mathfrak A$, then $E:\mathfrak A \to \mathcal A$ is said 
to be a \emph{conditional expectation with corresponding injection} $J$, if $E$ is completely positive
and $E \circ J =\operatorname{id}_{\mathcal A}$. Note that in this case $E$ is also unital (because $J$ is
unital) and thus of norm one. Hence the composed map $J\,\circ\, E:\mathfrak A\to J(\mathcal A)\subset\mathfrak A$ 
is a projection of norm one and thus a conditional expectation in the above sense. Thus every unital 
dilation gives rise to a conditional expectation $E$ with corresponding injection $J$.
\end{enumerate}
\end{remark}

Now Definition \ref{def:dilation-von-Neumann-algebras} directly applies to Heisenberg 
channels. Taking into account that the only invertible channels are the unitary ones
(cf.~Prop.~\ref{ch_3_Theorem_14}) we obtain the following concept.

\begin{defi}\label{def:dilation-Heisenberg-channles}
\begin{enumerate}
\item
Let $T\in Q_H(\mathcal H)$ be a Heisenberg quantum channel, i.e.~$T:\mathcal B(\mathcal H)\to\mathcal B(\mathcal H)$
is linear, ultraweakly continuous, completely positive and unital. A unitary dilation $(\mathcal{K},U,J,E)$ of
$T$ consists of a Hilbert space $\mathcal{K}$ and a triple of maps $(U,J,E)$ with 
\begin{equation*}
T = E \circ {\operatorname{Ad}_U} \circ J \quad \text{and} \quad E \circ J = \operatorname{id}_{\mathcal B(\mathcal H)}\,,
\end{equation*}
where $U$, $J$ and $E$ satisfy
\begin{enumerate}
\item 
$U\in\mathcal B(\mathcal{K})$ is unitary.
\item
$J:\mathcal B(\mathcal{H}) \to \mathcal B(\mathcal{K})$ is an ultraweakly continuous $*$-homomorphism of
$\mathcal B(\mathcal{H})$ into $\mathcal B(\mathcal{K})$.
\item
$E:\mathcal B(\mathcal{K})\to\mathcal B(\mathcal H)$ is linear, ultraweakly continuous and completely positive with operator norm $\|E\|=1$.
\end{enumerate}
\item
Let $S \subset G$ be a semigroup of a group $G$ and let $(T_g)_{g\in S}$ be a semigroup representation of $S$
with values in the set of Heisenberg quantum channels $Q_H(\mathcal H)$. A unitary dilation 
$(\mathcal{K},(U_g)_{g \in G_0}, J, E)$ of $(T_g)_{g \in S}$ consists of a Hilbert space $\mathcal{K}$, a subgroup
$G_0 \subset G$ and a triple $((U_g)_{g \in G_0}, J, E)$ with
\begin{equation*}
T_g = E \circ \operatorname{Ad}_{U_g} \circ J \quad \text{and} \quad E \circ J = \operatorname{id}_{\mathcal B(\mathcal H)}\,,
\end{equation*}
for all $g \in S\subset G_0$, where $(U_g)_{g \in G_0}$ is a representation of $G_0$ with values in the unitary group 
on $\mathcal K$ and $J,E$ as before.\vspace{4pt}
\end{enumerate}
If, in addition, $J(\operatorname{id}_{\mathcal H})=\operatorname{id}_{\mathcal K}$ then the dilation is said to be unital.
\end{defi}

\noindent 
If the dilation is unital, then $J\in Q_H(\mathcal H,\mathcal K)$ and $E\in Q_H(\mathcal K,\mathcal H)$ are Heisenberg channels (cf.~Rem.~\ref{rem_19}.2).
Finally, this concept can be transferred to the Schr\"odinger quantum channels via duality (cf.~Section
\ref{sec:dualchannel}).

\begin{defi}\label{def:dilation-Schroedinger-channles}
\begin{enumerate}
\item
Let $T\in Q_S(\mathcal H)$ be a Schr\"odinger quantum channel, i.e. $T:\mathcal B^1(\mathcal H)\to\mathcal B^1(\mathcal H)$ is linear, completely positive and trace-preserving. A unitary dilation $(\mathcal{K},U,J,E)$ of $T$ consists of
a Hilbert space $\mathcal{K}$ and a triple of maps $(U,J,E)$ with 
\begin{equation*}
T = E \circ {\operatorname{Ad}_U} \circ J \quad \text{and} \quad E \circ J = \operatorname{id}_{\mathcal B^1(\mathcal H)}\,,
\end{equation*}
where $U$, $J$ and $E$ satisfy
\begin{enumerate}
\item 
$U\in\mathcal B(\mathcal{K})$ is unitary.
\item
$J:\mathcal B^1(\mathcal{H})\to\mathcal B^1(\mathcal K)$ is linear and completely
positive with operator norm $\|J\|=1$. 
\item
$E:\mathcal B^1(\mathcal{K}) \to \mathcal B^1(\mathcal{H})$ is linear, completely positive and satisfies
\begin{equation}\label{eq:dual_star_hom}
E(E^*(B)A)=BE(A) \quad \text{for all $B\in\mathcal B(\mathcal H)$ and $A\in\mathcal B^1(\mathcal K)$,}
\end{equation}
where $E^*$ is the dual channel of $E$.
\end{enumerate}
\item
Let $S \subset G$ be a semigroup of a group $G$ and let $(T_g)_{g\in S}$ be a semigroup representation of $S$
with values in the set of Schr\"odinger quantum channels $Q_S(\mathcal H)$. A unitary dilation
$(\mathcal{K},(U_g)_{g \in G}, J, E)$ of $(T_g)_{g \in S}$ consists of a Hilbert space $\mathcal{K}$, a subgroup
$G_0 \subset G$ and a triple $((U_g)_{g \in G_0}, J, E)$ with
\begin{equation*}
T_g = E \circ \operatorname{Ad}_{U_g} \circ J \quad \text{and} \quad E \circ J = \operatorname{id}_{\mathcal B^1(\mathcal H)}\,,
\end{equation*}
for all $g \in S\subset G_0$, where $(U_g)_{g \in G_0}$ is a representation of $G_0$ with values in the unitary group 
on $\mathcal K$ and $J,E$ as before.\vspace{4pt}
\end{enumerate}
If, in addition, $E$ is trace-preserving then the dilation is said to be trace-preserving.
\end{defi}

\begin{remark}
\begin{enumerate}
\item 
Property \eqref{eq:dual_star_hom} which looks quite similar to \eqref{eq:cond_exp-1} implies (by direct 
computation) that the dual channel $E^*$ is a $*$-homomorphism. Moreover, $E^*$ is ultraweakly continuous
as this holds for every dual channel. Conversely, for any ultraweakly continuous $*$-homomorphism $J$ 
from Definition \ref{def:dilation-Heisenberg-channles} one can show that together with its 
pre-dual channel, it satisfies \eqref{eq:dual_star_hom}. In this sense, the dilation definitions \ref{def:dilation-Heisenberg-channles} and \ref{def:dilation-Schroedinger-channles} are dual to each other. Similar as for 
\eqref{eq:cond_exp-1}, one can conclude that \eqref{eq:dual_star_hom} is equivalent to
\begin{equation*}
E(AE^*(B))=E(A)B \quad \text{for all $B\in\mathcal B(\mathcal H)$ and $A\in\mathcal B^1(\mathcal K)$.}
\end{equation*}
\item
If a dilation is trace-preserving, then $J$ is trace-preserving as well (cf.~Remark \ref{rem_19}.2.) so in particular, $J\in Q_S(\mathcal H,\mathcal K)$ and $E\in Q_S(\mathcal K,\mathcal H)$ are Schr\"odinger channels.
\item
Corollary \ref{coro_1} shows that for every Heisenberg channel $T\in Q_H(\mathcal H)$ there exists a unitary (and even unital) dilation 
of $T$ of the following type $(\mathcal H\otimes\mathcal K,\operatorname{Ad}_U,i_{\mathcal K},\operatorname{tr}_\omega)$,
where $\mathcal K$ is a separable Hilbert space, $\omega\in\mathbb D(\mathcal K)$ a pure state and 
$U\in\mathcal B(\mathcal H\otimes\mathcal K)$ a unitary operator. Such a dilation is also said to be of tensor type.
This result holds analogously for every $T\in Q_S(\mathcal H)$ by Theorem \ref{thm1}.
\end{enumerate}
\end{remark}

\bibliography{article_vomende_dirr_unitary_dilations_update_2019_03_11}
\end{document}